\documentclass[10pt,letterpaper]{article}
\usepackage[top=1in,left=1in,footskip=0.75in,marginparwidth=2in]{geometry}

\usepackage[utf8]{inputenc}

\usepackage[authoryear,round]{natbib}
\usepackage{nameref}
\usepackage{varioref}
\usepackage{hyperref}
\usepackage{cleveref}

\usepackage{microtype}
\DisableLigatures[f]{encoding = *, family = * }

\usepackage{changepage}

\usepackage[aboveskip=1pt,labelfont=bf,labelsep=period,singlelinecheck=off]{caption}

\usepackage{lastpage,fancyhdr,graphicx}
\usepackage{epstopdf}
\usepackage{color}

\definecolor{Gray}{gray}{.25}

\usepackage{graphicx}

\usepackage{sidecap}

\usepackage{wrapfig}
\usepackage[pscoord]{eso-pic}
\usepackage[fulladjust]{marginnote}
\reversemarginpar

\usepackage{amsmath,amsthm}
\usepackage{bm}
\newtheorem{thm}{Theorem}
\newtheorem{lem}{Lemma}

\usepackage[font=small,labelfont=bf]{caption}

\begin{document}
\vspace*{0.35in}

\begin{flushleft}
{\Large
\textbf\newline{Towards deep learning with segregated dendrites}
}
\newline

Jordan Guergiuev\textsuperscript{1,2},
Timothy P. Lillicrap\textsuperscript{4}, and 
Blake A. Richards\textsuperscript{1,2,3,*}
\\
\bigskip
\bf{1} Department of Biological Sciences, University of Toronto Scarborough, Toronto, ON, Canada
\\
\bf{2} Department of Cell and Systems Biology, University of Toronto, Toronto, ON, Canada
\\
\bf{3} Learning in Machines and Brains Program, Canadian Institute for Advanced Research, Toronto, ON, Canada
\\
\bf{4} DeepMind Technologies Inc., London, UK
\\
\bigskip
* Corresponding author, email: blake.richards@utoronto.ca

\end{flushleft}

\section*{Abstract}
\label{abstract}
Deep learning has led to significant advances in artificial intelligence, in part, by adopting strategies motivated by neurophysiology. However, it is unclear whether deep learning could occur in the real brain. Here, we show that a deep learning algorithm that utilizes multi-compartment neurons might help us to understand how the brain optimizes cost functions. Like neocortical pyramidal neurons, neurons in our model receive sensory information and higher-order feedback in electrotonically segregated compartments. Thanks to this segregation, the neurons in different layers of the network can coordinate synaptic weight updates. As a result, the network can learn to categorize images better than a single layer network. Furthermore, we show that our algorithm takes advantage of multilayer architectures to identify useful representations---the hallmark of deep learning. This work demonstrates that deep learning can be achieved using segregated dendritic compartments, which may help to explain the dendritic morphology of neocortical pyramidal neurons.

\section*{Introduction}
\label{introduction}
Deep learning refers to an approach in artificial intelligence (AI) that utilizes neural networks with multiple layers of processing units. Importantly, deep learning algorithms are designed to take advantage of these multi-layer network architectures in order to generate hierarchical representations wherein each successive layer identifies increasingly abstract, relevant variables for a given task \citep{bengio_scaling_2007,lecun_deep_2015}. In recent years, deep learning has revolutionized machine learning, opening the door to AI applications that can rival human capabilities in pattern recognition and control \citep{mnih_human-level_2015,silver_mastering_2016,he_delving_2015}. Interestingly, the representations that deep learning generates resemble those observed in the neocortex, particularly in higher-order sensory areas \citep{kubilius_deep_2016,khaligh-razavi_deep_2014,cadieu_deep_2014}, suggesting that something akin to deep learning is occurring in the mammalian brain \citep{yamins_using_2016,marblestone2016towards}.
\par
Yet, a large gap exists between deep learning in AI and our current understanding of learning and memory in neuroscience. In particular, unlike deep learning researchers, neuroscientists do not yet have a solution to the ``credit assignment problem'' \citep{rumelhart_learning_1986,lillicrap_random_2014,bengio_towards_2015}. The credit assignment problem refers to the fact that the behavioral consequences of synaptic changes in early layers of a neural network depend on the current connections in the downstream layers. For example, consider the behavioral effects of synaptic changes, i.e. long-term potentiation/depression (LTP/LTD), occurring between different sensory circuits of the brain. Exactly how these synaptic changes will impact behavior and cognition depends on the downstream connections between the sensory circuits and motor or associative circuits (\hyperref[fig:F1]{Figure 1A}). Hence, learning to optimize some behavioral or cognitive function requires a method for identifying what a given neuron's place is within the wider neural network, i.e. it requires a means of assigning ``credit'' (or ``blame'') to neurons in early sensory areas for their contribution to the final behavioral output \citep{lecun_deep_2015,bengio_towards_2015}. Despite its importance for real-world learning, the credit assignment problem has received little attention in neuroscience.
\par
The lack of attention to credit assignment in neuroscience is, arguably, a function of the history of biological studies of synaptic plasticity. Due to the well-established dependence of LTP and LTD on presynaptic and postsynaptic activity, current theories of learning in neuroscience tend to emphasize Hebbian learning algorithms \citep{dan_spike_2004,martin_synaptic_2000}, that is, learning algorithms where synaptic changes depend solely on presynaptic and postsynaptic activity. Hebbian learning models can produce representations that resemble the representations in the real brain \citep{zylberberg_sparse_2011,leibo_view-tolerant_2017} and they are backed up by decades of experimental findings \citep{malenka_ltp_2004,dan_spike_2004,martin_synaptic_2000}. But, current Hebbian learning algorithms do not solve the credit assignment problem, nor do global neuromodulatory signals used in reinforcement learning \citep{lillicrap_random_2014}. As a result, deep learning systems from AI that can do credit assignment outperform existing Hebbian models of perceptual learning on a variety of tasks \citep{yamins_using_2016,khaligh-razavi_deep_2014}. This suggests that a critical, missing component in our current models of the neurobiology of learning and memory is an explanation of how the brain solves the credit assignment problem. 
\par
However, the most common solution to the credit assignment problem in AI is to use the backpropagation of error algorithm \citep{rumelhart_learning_1986}. Backpropagation assigns credit by \textit{explicitly} using current downstream synaptic connections to calculate synaptic weight updates in earlier layers, commonly termed ``hidden layers'' \citep{lecun_deep_2015} (\hyperref[fig:F1]{Figure 1B}). This technique, which is sometimes referred to as ``weight transport'', involves non-local transmission of synaptic weight information between layers of the network \citep{lillicrap_random_2014,grossberg_competitive_1987}. Weight transport is clearly unrealistic from a biological perspective \citep{bengio_towards_2015,crick_recent_1989}. It would require early sensory processing areas (e.g. V1, V2, V4) to have precise information about \textit{billions} of synaptic connections in downstream circuits (MT, IT, M2, EC, etc.). According to our current understanding, there is no physiological mechanism that could communicate this information in the brain. Some deep learning algorithms utilize purely Hebbian rules \citep{scellier2016towards,hinton_fast_2006}. But they depend on feedback synapses that are symmetric with feedforward synapses to solve the credit assignment problem \citep{scellier2016towards,hinton_fast_2006}, which is essentially a version of weight transport. Altogether, these artificial aspects of current deep learning solutions to credit assignment have rendered many scientists skeptical to the proposal that deep learning occurs in the real brain \citep{crick_recent_1989,grossberg_competitive_1987,harris_stability_2008,urbanczik_reinforcement_2009}.
\par
Recent findings have shown that these problems may be surmountable, though. \cite{lillicrap_random_2014}, \cite{lee_target_2014} and \cite{liao_how_2015} have demonstrated that it is possible to solve the credit assignment problem even while avoiding weight transport or symmetric feedback weights. The key to these learning algorithms is the use of feedback signals from output layers that are transmitted to calculate a local error signal, which then guides synaptic updates in hidden layers \citep{lee_target_2014,lillicrap_random_2014,liao_how_2015}. These learning algorithms can take advantage of multi-layer architectures, leading to performance that rivals backpropagation \citep{lee_target_2014,lillicrap_random_2014,liao_how_2015}. Hence, this work has provided a significant breakthrough in our understanding of how the real brain might do credit assignment.
\par
Nonetheless, the models of \cite{lillicrap_random_2014}, \cite{lee_target_2014} and \cite{liao_how_2015} involve some problematic assumptions. Specifically, although it is not directly stated in all of the papers, there is an implicit assumption that there is a separate feedback pathway for transmitting the signals that drive synaptic updates in hidden layers (\hyperref[fig:F2]{Figure 2A}). Such a pathway is required in these models because the synaptic weight updates in the hidden layers depend on the difference between feedback that is generated in response to a purely feedforward propagation of sensory information, and feedback that is guided by a teaching signal \citep{lillicrap_random_2014, lee_target_2014, liao_how_2015}. In order to calculate this difference, sensory information must be transmitted \textit{separately} from the feedback signals that are used to drive learning. In single compartment neurons, keeping feedforward sensory information separate from feedback signals is impossible without a separate pathway. At face value, such a pathway is possible. But, closer inspection uncovers a couple of difficulties with such a proposal.
\par
First, the error signals that solve the credit assignment problem are not global error signals (like neuromodulatory signals used in reinforcement learning). Rather, they are \textit{cell-by-cell} error signals. This would mean that the feedback pathway would require some degree of pairing, wherein each neuron in the hidden layer is paired with a feedback neuron (or circuit). That is not impossible, but there is no evidence to date of such an architecture in the neocortex. Second, the error that is communicated to the hidden layer is signed (i.e. it can be positive or negative), and the sign determines whether LTP or LTD occur in the hidden layer neurons \citep{lee_target_2014,lillicrap_random_2014,liao_how_2015}. Communicating signed signals with a spiking neuron can theoretically be done by using a baseline firing rate that the neuron can go above (for positive signals) or below (for negative signals). But, in practice, such systems are difficult to operate because as the error gets closer to zero any noise in the spiking of the neuron can switch the sign of the signal, which switches LTP to LTD, or \textit{vice versa}. This means that as learning progresses the network's ability to communicate error signs gets \textit{worse}. Therefore, the real brain's specific solution to the credit assignment problem is unlikely to involve a separate feedback pathway for cell-by-cell, signed signals to instruct plasticity.
\par
However, segregating the integration of feedforward and feedback signals does not require a separate pathway if neurons have more complicated morphologies than the point neurons typically used in artificial neural networks. Taking inspiration from biology, we note that real neurons are much more complex than single-compartments, and different signals can be integrated at distinct dendritic locations. Indeed, in the primary sensory areas of the neocortex, feedback from higher-order areas arrives in the distal apical dendrites of pyramidal neurons \citep{manita_top-down_2015,budd_extrastriate_1998,spratling_cortical_2002}, which are electrotonically very distant from the basal dendrites where feedforward sensory information is received \citep{larkum_new_1999,larkum_dendritic_2007,larkum_synaptic_2009}. Thus, as has been noted by previous authors \citep{kording_supervised_2001,spratling_cortical_2002,spratling_feedback_2006}, the anatomy of pyramidal neurons may actually provide the necessary segregation of feedforward and feedback information to calculate local error signals and perform deep learning in biological neural networks.
\par
Here, we show how deep learning can be implemented if neurons in hidden layers contain segregated ``basal'' and ``apical'' dendritic compartments for integrating feedforward and feedback signals separately (\hyperref[fig:F2]{Figure 2B}). Our model builds on previous neural networks research \citep{lee_target_2014,lillicrap_random_2014} as well as computational studies of supervised learning in multi-compartment neurons \citep{urbanczik_learning_2014,kording_supervised_2001,spratling_feedback_2006}. Importantly, we use the distinct basal and apical compartments in our neurons to integrate feedback signals separately from feedforward signals. With this, we build a local target for each hidden layer that coordinates learning across the layers of the network. We demonstrate that even with random synaptic weights for feedback into the apical compartment, our algorithm can coordinate learning to achieve classification of the MNIST database of hand-written digits better than can be achieved with single layer networks. Furthermore, we show that our algorithm allows the network to take advantage of multi-layer structures to build hierarchical, abstract representations, one of the hallmarks of deep learning \citep{lecun_deep_2015}. Our results demonstrate that deep learning can be implemented in a biologically feasible manner if feedforward and feedback signals are received at electrotonically segregated dendrites, as is the case in the mammalian neocortex.

\section*{Results}
\label{results}

\subsection*{A network architecture with segregated dendritic compartments}

Deep supervised learning with local weight updates requires that each neuron receive signals that can be used to determine its ``credit'' for the final behavioral output. We explored the idea that the cortico-cortical feedback signals to pyramidal cells could provide the required information for credit assignment. In particular, we were inspired by four observations from both machine learning and biology:

\begin{enumerate}
\item Current solutions to credit assignment without weight transport require segregated feedforward and feedback signals \citep{lee_target_2014,lillicrap_random_2014}.\\

\item In the neocortex, feedforward sensory information and higher-order cortico-cortical feedback are largely received by distinct dendritic compartments, namely the basal dendrites and distal apical dendrites, respectively \citep{spratling_cortical_2002,budd_extrastriate_1998}. \\

\item The distal apical dendrites of pyramidal neurons are electrotonically distant from the soma, such that passive transmission from the distal dendrites to the soma is significantly attenuated. Instead, apical communication to the soma depends on active propagation through the apical dendritic shaft driven by voltage-gated calcium channels. These non-linear, active events in the apical shaft generate prolonged upswings in the membrane potential, known as ``plateau potentials'', which can drive burst firing at the soma \citep{larkum_new_1999,larkum_synaptic_2009}.

\item Plateau potentials driven by apical activity can guide plasticity \textit{in vivo} \citep{bittner_conjunctive_2015}.

\end{enumerate}

With these considerations in mind, we hypothesized that the computations required for credit assignment could be achieved without any separate pathways for feedback signals. Instead, they could be achieved by having two distinct dendritic compartments in each hidden layer neuron: a ``basal'' compartment, strongly coupled to the soma for integrating feedforward information, and an ``apical'' compartment for integrating feedback information that would only drive activity at the soma when ``plateau potentials'' occur (\hyperref[fig:F3]{Figure 3A}). 
\par
As an initial test of this concept we built a network with a single hidden layer. Although this network is not very ``deep'', even a single hidden layer can improve performance over a one-layer architecture if the learning algorithm solves the credit assignment problem \citep{bengio_scaling_2007, lillicrap_random_2014}. Hence, we wanted to initially determine whether our network could take advantage of a hidden layer to reduce error at the output layer.
\par 
The network architecture is illustrated in \hyperref[fig:F3]{Figure 3A}. An image from the MNIST data set is used to set the spike rates of $\ell = 784$ Poisson point-process neurons in the input layer (one neuron per image pixel, rates-of-fire determined by pixel intensity). These project to a hidden layer with $m=500$ neurons. The neurons in the hidden layer are composed of three distinct compartments with their own voltages: the apical compartments (with voltages described by the vector $\bm{A}(t) = [A_1(t), ..., A_m(t)]$), the basal compartments (with voltages $\bm{B}(t) = [B_1(t), ..., B_m(t)]$), and the somatic compartments (with voltages $\bm{C}(t) = [C_1(t), ..., C_m(t)]$). (\textit{Note}: for notational clarity, all vectors and matrices in the paper are in boldface.) The voltages in the dendritic compartments are calculated as weighted sums of a postsynaptic potential kernel convolved with the incoming spike train (see \nameref{methods}, equations \eqref{eqn:psp} and \eqref{eqn:dend_voltages_hidden}), while the somatic voltages are calculated as leaky integrators of the dendritic inputs, i.e. for the $i^{th}$ hidden layer neuron:

\begin{align}
\label{eqn:soma_voltage_hidden_apical_conductance}
\frac{dC_i(t)}{dt} &= -g_L C_i(t) + g_B(B_i(t) - C_i(t)) + g_A(A_i(t) - C_i(t))
\end{align}

where $g_L$, $g_B$ and $g_A$ represent the leak conductance, the conductance from the basal dendrites, and the conductance from the apical dendrites, respectively. Note, for mathematical simplicity we are assuming a resting membrane potential of 0 V and a membrane capacitance of 1 F (these values do not affect the results). We implement electrotonic segregation in the model by altering the $g_A$ value---low values for $g_A$ lead to electrotonically segregated apical dendrites. In the initial set of simulations we set $g_A=0$, but we relax this hard constraint in later simulations. The somatic compartments generate spikes using Poisson processes. The instantaneous rates of these processes are described by the vector $\bm{\lambda}^C(t) = [\lambda_1^C(t), ..., \lambda_m^C(t)]$, which is in units of spikes/s or Hz. These rates-of-fire are determined by a non-linear sigmoid function, $\sigma(\cdot)$, applied to the somatic voltages, i.e. for the $i^th$ hidden layer neuron:

\begin{align}
\begin{split}
\label{eqn:spike_rates}
\lambda_i^C(t) &= \lambda_{max} \sigma(C_i(t)) \\
               &= \lambda_{max} \frac{1}{1 + e^{-C_i(t)}}
\end{split}
\end{align}

where $\lambda_{max}$ is the maximum rate-of-fire for the neurons.
\par
Spiking inputs from the input layer arrive at the basal compartments via the $m \times \ell$ synaptic weight matrix $\bm{W}^0$, hidden layer somatic spikes are projected to the output layer neurons via the $n \times m$ synaptic weight matrix $\bm{W}^1$, and spiking inputs from the output layer arrive at the apical compartments via the $m \times n$ synaptic weight matrix $\bm{Y}$ (\hyperref[fig:F3]{Figure 3A}).
\par
The output layer neurons consist of $n=10$ two compartment neurons (one for each image category), similar to those used in a previous model of dendritic prediction learning \citep{urbanczik_learning_2014}. The output dendritic voltages ($\bm{V}(t) = [V_1(t), ..., V_n(t)]$) and somatic voltages ($\bm{U}(t) = [U_1(t), ..., U_n(t)]$) are updated in a similar manner to the hidden layer basal compartment and soma, respectively (see \nameref{methods}, equations \eqref{eqn:dend_voltage_output} and \eqref{eqn:soma_voltage_output}). As well, like the hidden layer neurons, the output neurons spike using Poisson processes whose instantaneous rates, $\bm{\lambda}^U(t) = [\lambda_1^U(t), ..., \lambda_n^U(t)]$, are determined by the somatic voltages, i.e. $\lambda_i^U(t) = \lambda_{max} \sigma(U_i(t))$. Unlike the hidden layer, however, the output layer neurons also receive a set of ``teaching signals'', specifically, excitatory and inhibitory conductances that push their voltages towards target values (see \nameref{methods}, equations \eqref{eqn:soma_voltage_output} and \eqref{eqn:somatic_current}). These teaching signals are similar to the targets that are used for training in deep artificial neural networks \citep{bengio_scaling_2007, lecun_deep_2015}. Whether any such teaching signals exist in the real brain is unknown, though there is evidence that animals can represent desired behavioral outputs with internal goal representations \citep{gadagkar_dopamine_2016}.
\par
Critically, we define two different modes of integration in the hidden layer neurons: ``transmit'' and ``plateau''. During the transmit mode, the apical compartment is considered to be electrotonically segregated from the soma so it affects the somatic voltage minimally (depending on $g_A$), and most of the conductance to the soma comes from the basal compartment (\hyperref[fig:F3]{Figure 3B}, left). In contrast, during the plateau mode, the apical voltage is averaged over the most recent 20-30 ms period and the sigmoid non-linearity is applied to it, giving us ``plateau potentials'' (see equation \eqref{eqn:plateau_potential} below). These non-linear versions of the apical voltages are then transmitted to the somatic and basal compartments for synaptic updates (\hyperref[fig:F3]{Figure 3B}, right). The intention behind this design was to mimic the non-linear transmission from the apical dendrites to the soma that occurs during a plateau potential driven by calcium spikes in the apical dendritic shaft \citep{larkum_new_1999}. Furthermore, the temporal averaging was intended to mimic, in part, the temporal dynamics introduced by NMDA plateaus \citep{schiller_nmda_2000}, which are particularly good for driving apical calcium spikes \citep{larkum_synaptic_2009}.
\par
To train the network we alternate between two phases. First, we present an image to the input layer without any signals to the output layer during the ``forward'' phase, which occurs between times $t_0$ to $t_1$. At $t_1$ a plateau potential occurs in all the hidden layer neurons and the ``target'' phase begins. During this phase the image continues to drive the input layer, but now the output layer also receives teaching conductances, which force the output units closer to the correct answer. For example, if an image of a `9' is presented, then over the time period $t_1$-$t_2$ the `9' neuron in the output layer receives strong excitatory conductances from the teaching signal, while the other neurons receive strong inhibitory conductances (\hyperref[fig:F3]{Figure 3C}). This phase lasts from $t_1$ to $t_2$, and at $t_2$ another set of plateau potentials occurs in the hidden layer neurons. The result is that we have plateaus potentials in the hidden layer neurons for both the forward ($\bm{\alpha}^f = [\alpha^f_1, ..., \alpha^f_m]$) and target ($\bm{\alpha}^t = [\alpha^t_1, ..., \alpha^t_m]$) phases: 

\begin{align}
\begin{split}
\label{eqn:plateau_potential}
\alpha^f_i &= \sigma(\int_{t_1-t_p}^{t_1} A_i(t) dt) \\
\alpha^t_i &= \sigma(\int_{t_2-t_p}^{t_2} A_i(t) dt) \\
\end{split}
\end{align}

where $t_p$ is the integration time for the plateau potential (see \nameref{methods}). Note that since the plateau potentials are generated using the same sigmoid non-linearity that we use to calculate the rate-of-fire at the soma, they are analogous to the firing rates of the hidden layer neurons. This allows us to use the plateau potentials to define target firing rates for the neurons (see below). The network is simulated in near continuous-time (except that each plateau is considered to be instantaneous), and the intervals between plateaus are randomly sampled from an inverse Gaussian distribution (\hyperref[fig:F3]{Figure 3D}, top). As such, the specific amount of time that the network is presented with each image and teaching signal is stochastic, though usually somewhere between 50-60 ms of simulated time (\hyperref[fig:F3]{Figure 3D}, bottom). This stochasticity was not necessary, but it demonstrates that although the system operates in phases, the specific length of the phases is not important as long as they are sufficiently long to permit integration (see \hyperref[lem1]{Lemma 1}). In the data presented in this paper, all 60,000 images in the MNIST training set were presented to the network one at a time, and each exposure to the full set of images was considered an ``epoch'' of training. At the end of each epoch, the network's classification error rate on a separate set of 10,000 test images was assessed with a single forward phase (see \nameref{methods}).
\par
It is important to note that there are many aspects of this design that are not physiologically accurate. Most notably, stochastic generation of synchronized plateau potentials across a population is not an accurate reflection of how real pyramidal neurons operate, since apical calcium spikes are determined by a number of concrete physiological factors in individual cells, including back-propagating action potentials, spike-timing and inhibitory inputs \citep{larkum_new_1999,larkum_dendritic_2007,larkum_synaptic_2009}. However, we note that calcium spikes in the apical dendrites can be prevented from occurring via the activity of distal dendrite targeting inhibitory interneurons \citep{murayama_dendritic_2009}, which can synchronize pyramidal activity \citep{hilscher_chrna2-martinotti_2017}. Furthermore, distal dendrite targeting neurons can themselves can be rapidly inhibited in response to temporally precise neuromodulatory inputs \citep{pi_cortical_2013,pfeffer_inhibition_2013,karnani_opening_2016,hangya_central_2015,brombas_activity-dependent_2014}. Therefore, it is entirely plausible that neocortical micro-circuits would generate synchronized pyramidal plateaus at punctuated periods of time in response to disinhibition of the apical dendrites governed by neuromodulatory signals that determine ``phases'' of processing. Alternatively, oscillations in population activity could provide a mechanism for promoting alternating phases of processing and synaptic plasticity \citep{buzsaki_neuronal_2004}. But, complete synchrony of plateaus in our hidden layer neurons is not actually critical to our algorithm---only the temporal relationship between the plateaus and the teaching signal is critical (see below). This relationship itself is arguably plausible given the role of neuromodulatory inputs in dis-inhibiting the distal dendrites of pyramidal neurons \citep{karnani_opening_2016,brombas_activity-dependent_2014}. Of course, we are engaged in a great deal of speculation here. But, the point is that our model utilizes anatomical and functional motifs that are analogous to what is observed in the neocortex. Importantly for the present study, the key issue is the use of segregated dendrites and distinct transmit and plateau modes to solve the credit assignment problem.

\subsection*{Credit assignment with segregated dendrites}

To solve the credit assignment problem without using weight transport, we had to define a local error function for the hidden layer that somehow takes into account the impact that each hidden layer neuron has on the output at the final layer. In other words, credit assignment could only be achieved if we knew that changing a hidden layer synapse to reduce the local error would also help to reduce error at the output layer. To obtain this guarantee, we defined local targets for the output and the hidden layer, i.e. desired firing rates for both the output layer neurons and the hidden layer neurons. Learning is then a process of changing the synaptic connections to achieve these target firing rates across the network. Importantly, we defined our hidden layer targets in a similar manner to \cite{lee_target_2014}. These hidden layer targets coordinate with the output layer by incorporating the feedback information contained in the plateau potentials. In this way, credit assignment is accomplished through mechanisms that are spatially local to the hidden layer neurons.
\par
Specifically, to create the output and hidden layer targets, we use the average rates-of-fire and average voltages during different phases. We define:

\begin{align}
\begin{split}
\label{eqn:average_vectors}
\overline{\bm{X}}^y &= [\overline{X_1}^y, ..., \overline{X_k}^y] \\
\overline{X_i}^f &= \int_{t_1-t_p}^{t_1} X_i(t) dt \\
\overline{X_i}^t &= \int_{t_2-t_p}^{t_2} X_i(t) dt \\
\end{split}
\end{align}

where $y \in \{f,t\}$, $k \in \{\ell,m,n\}$ and $X_i(t)$ can be a voltage variable (e.g. $C_i(t)$ or $U_i(t)$) or a rate-of-fire variable (e.g. $\lambda^C_i(t)$ or $\lambda^U_i(t)$). Using the averages over different phases, we then define the target rates-of-fire for the output layer, $\hat{\bm{\lambda}}^U = [\hat{\lambda}_1^U, ..., \hat{\lambda}_n^U]$, to be the average rates-of-fire during the target phase:

\begin{align}
\label{eqn:output_target}
\hat{\lambda}_i^U &= \overline{\lambda^U_i}^{t} 
\end{align}

For the hidden layer we define the target rates-of-fire, $\hat{\bm{\lambda}}^C = [\hat{\lambda}_1^C, ..., \hat{\lambda}_m^C]$, using the average rates-of-fire during the forward phase and the difference between the plateau potentials from the forward and transmit phase:

\begin{align}
\label{eqn:hidden_target}
\hat{\lambda}_i^C &= \overline{\lambda_i^C}^f + \alpha_i^t - \alpha_i^f
\end{align}

The goal of learning in the hidden layer is to change the synapses $\bm{W}^0$ to achieve these targets in response to the given inputs.
\par 
More generally, for both the output layer and the hidden layer we then define error functions, $L^1$ and $L^0$, respectively, based on the difference between the local targets and the activity of the neurons during the forward phase (\hyperref[fig:F4]{Figure 4A}):

\begin{align}
\label{eqn:loss_functions}
\begin{split}
L^1 &= || \hat{\bm{\lambda}}^U - \lambda_{max} \sigma(\overline{\bm{U}}^f) ||^2_2 \\
L^0 &= || \hat{\bm{\lambda}}^C - \lambda_{max} \sigma(\overline{\bm{C}}^f) ||^2_2
\end{split}
\end{align}

Note that the loss function for the hidden layer, $L^0$, will, on average, reduce to the difference $||\bm{\alpha}^t - \bm{\alpha}^f||^2_2$. This provides an intuitive reason for why these targets help with credit assignment. Specifically, a hidden layer neuron's error is small only when the output layer is providing similar feedback to it during the forward and target phases. Put another way, hidden layer neurons ``know'' that they are ``doing well'' when they receive the same feedback from the output layer regardless of whether or not the teaching signal is present. Thus, hidden layer neurons have access to some information about how they are doing in helping the output layer to achieve its targets.
\par
More formally, it can be shown that the output and hidden layer error functions will, on average, agree with each other, i.e. if the hidden layer error is reduced then the output layer error is also reduced. Specifically, similar to the proof employed by \cite{lee_target_2014}, it can be shown that if the output layer error is sufficiently small and the synaptic matrices $\bm{W}^1$ and $\bm{Y}$ meet some conditions, then: 

\begin{align}
\label{eqn:coordination_condition}
|| \hat{\bm{\lambda}}^U - \lambda_{max}\sigma(k_D \bm{W}^1 \hat{\bm{\lambda}}^C) ||_2^2 < || \hat{\bm{\lambda}}^U - \lambda_{max}\sigma(E[\overline{\bm{U}}^f]) ||_2^2
\end{align}

where $k_D$ is a conductance term and $E[\cdot]$ denotes the expected value. (See \hyperref[thm1]{Theorem 1} for the proof and a concrete description of the conditions). In plain language, equation \eqref{eqn:coordination_condition} says that when we consider the difference between the output target and the activity that the hidden target \textit{would} have induced at the output layer, it is less than the difference between the output target and the expected value of the output layer activity during the forward phase. In other words, the output layer's error would have, on average, been \textit{smaller} if the hidden layer had achieved its own target during the forward phase. Hence, if we reduce the hidden layer error, it should also reduce the output layer error, thereby providing a guarantee of appropriate credit assignment.
\par
With these error functions, we update the weights in the network with local gradient descent:

\begin{align}
\label{eqn:updates}
\begin{split}
\Delta \bm{W}^1 \propto \frac{\partial L^1}{\partial \bm{W}^1} \\
\Delta \bm{W}^0 \propto \frac{\partial L^0}{\partial \bm{W}^0}
\end{split}
\end{align}

where $\Delta \bm{W}^i$ refers to the update term for weight matrix $\bm{W}^i$ (see \nameref{methods}, equations \eqref{eqn:error_signal_output}, \eqref{eqn:weight_update_output}, \eqref{eqn:error_signal_hidden} and \eqref{eqn:weight_update_hidden} for details of the weight update procedures). Given the coordination implied by equation \eqref{eqn:coordination_condition}, as the hidden layer reduces its own local error with gradient descent, the output layer's error should also be reduced, i.e. hidden layer learning should imply output layer learning.
\par
To test that we were successful in credit assignment for the hidden layer, and to provide empirical support for the proof, we compared the local error at the hidden layer to the output layer error across all of the image presentations to the network. We observed that, generally, whenever the hidden layer error was low, the output layer error was also low. For example, when we consider the errors for the set of `2' images presented to the network during the second epoch, there was a Pearson correlation coefficient between $L^0$ and $L^1$ of $r=0.61$, which was much higher than what was observed for shuffled data, wherein output and hidden activities were randomly paired (\hyperref[fig:F4]{Figure 4B}). Furthermore, these correlations were observed across all epochs of training, with most correlation coefficients for the hidden and output errors falling between $r=0.2$ - $0.6$, which was, again, much higher than the correlations observed for shuffled data (\hyperref[fig:F4]{Figure 4C}). 
\par
Interestingly, the correlations between $L^0$ and $L^1$ were smaller on the first epoch of training. This suggests that the coordination between the layers may only come into full effect once the network has engaged in some learning. Therefore, we inspected whether the conditions on the synaptic matrices that are assumed in the proof were, in fact, being met. More precisely, the proof assumes that the feedforward and feedback synaptic matrices ($\bm{W}^1$ and $\bm{Y}$, respectively) produce forward and backward transformations between the output and hidden layer that are approximate inverses of each other (see Proof of \hyperref[thm1]{Theorem 1}). Since we begin learning with random matrices, this condition is almost definitely \textit{not} met at the start of training. But, we found that the network learned to meet this condition. Inspection of $\bm{W}^1$ and $\bm{Y}$ showed that during the first epoch the forward and backwards functions became approximate inverses of each other (\hyperref[fig:F4S1]{Figure 4, Supplement 1}). This means that during the first few image presentations the network was actually \textit{learning to do credit assignment}. This result is very similar to previous models examining feedback alignment \citep{lillicrap_random_2014}, and shows that the feedback alignment \citep{lillicrap_random_2014} and difference target propagation \citep{lee_target_2014} algorithms are intimately linked. Furthermore, our results suggest that very early development may involve a period of learning how to assign credit appropriately. Altogether, our model demonstrates that credit assignment using random feedback weights is a general principle that can be implemented using segregated dendrites.

\subsection*{Deep learning with segregated dendrites}
Given our finding that the network was successfully assigning credit for the output error to the hidden layer neurons, we had reason to believe that our network with local weight-updates would exhibit deep learning, i.e. an ability to take advantage of a multi-layer structure \citep{bengio_scaling_2007}. To test this, we examined the effects of including hidden layers. If deep learning is indeed operational in the network, then the inclusion of hidden layers should improve the ability of the network to classify images.
\par
We built three different versions of the network (\hyperref[fig:F5]{Figure 5A}). The first was a network that had no hidden layer, i.e. the input neurons projected directly to the output neurons. The second was the network illustrated in \hyperref[fig:F3]{Figure 3A}, with a single hidden layer. The third contained two hidden layers, with the output layer projecting directly back to both hidden layers. This direct projection allowed us to build our local targets for each hidden layer using the plateaus driven by the output layer, thereby avoiding a ``backward pass'' through the entire network as has been used in other models \citep{lillicrap_random_2014,lee_target_2014,liao_how_2015}. We trained each network on the 60,000 MNIST training images for 60 epochs, and recorded the percentage of images in the 10,000 test image set that were incorrectly classified. The network with no hidden layers rapidly learned to classify the images, but it also rapidly hit an asymptote at an average error rate of 8.3\% (\hyperref[fig:F5]{Figure 5B}, gray line). In contrast, the network with one hidden layer did not exhibit a rapid convergence to an asymptote in its error rate. Instead, it continued to improve throughout all 60 epochs, achieving an average error rate of 4.1\% by the 60\textsuperscript{th} epoch (\hyperref[fig:F5]{Figure 5B}, blue line). Similar results were obtained when we loosened the synchrony constraints and instead allowed each hidden layer neuron to engage in plateau potentials at different times (\hyperref[fig:F5S1]{Figure 5, Supplement 1}). This demonstrates that strict synchrony in the plateau potentials is not required. But, our target definitions do require two different plateau potentials separated by the teaching signal input, which mandates some temporal control of plateau potentials in the system.
\par
Interestingly, we found that the addition of a second hidden layer further improved learning. The network with two hidden layers learned more rapidly than the network with one hidden layer and achieved an average error rate of 3.2\% on the test images by the 60\textsuperscript{th} epoch, also without hitting a clear asymptote in learning (\hyperref[fig:F5]{Figure 5B}, red line). However, it should be noted that additional hidden layers beyond two did not significantly improve the error rate (data not shown), which suggests that our particular algorithm could not be used to construct very deep networks as is. Nonetheless, our network was clearly able to take advantage of multi-layer architectures to improve its learning, which is the key feature of deep learning \citep{bengio_scaling_2007, lecun_deep_2015}.
\par
Another key feature of deep learning is the ability to generate representations in the higher layers of a network that capture task-relevant information while discarding sensory details \citep{lecun_deep_2015,mnih_human-level_2015}. To examine whether our network exhibited this type of abstraction, we used the t-Distributed Stochastic Neighbor Embedding algorithm (t-SNE). The t-SNE algorithm reduces the dimensionality of data while preserving local structure and non-linear manifolds that exist in high-dimensional space, thereby allowing accurate visualization of the structure of high-dimensional data \citep{maaten_visualizing_2008}. We applied t-SNE to the activity patterns at each layer of the two hidden layer network for all of the images in the test set after 60 epochs of training. At the input level, there was already some clustering of images based on their categories. However, the clusters were quite messy, with different categories showing outliers, several clusters, or merged clusters (\hyperref[fig:F5]{Figure 5C}, bottom). For example, the `2' digits in the input layer exhibited two distinct clusters separated by a cluster of `7's: one cluster contained `2's with a loop and one contained `2's without a loop. Similarly, there were two distinct clusters of `4's and `9's that were very close to each other, with one pair for digits on a pronounced slant and one for straight digits (\hyperref[fig:F5]{Figure 5C}, bottom, example images). Thus, although there is built-in structure to the categories of the MNIST dataset, there are a number of low-level features that do not respect category boundaries. In contrast, at the first hidden layer, the activity patterns were much cleaner, with far fewer outliers and split/merged clusters (\hyperref[fig:F5]{Figure 5C}, middle). For example, the two separate `2' digit clusters were much closer to each other and were now only separated by a very small cluster of `7's. Likewise, the `9' and `4' clusters were now distinct and no longer split based on the slant of the digit. Interestingly, when we examined the activity patterns at the second hidden layer the categories were even better segregated with only a bit of splitting or merging of category clusters (\hyperref[fig:F5]{Figure 5C}, top). Therefore, the network had learned to develop representations in the hidden layers wherein the categories were very distinct and low-level features unrelated to the categories were largely ignored. This abstract representation is likely to be key to the improved error rate in the two hidden layer network. Altogether, our data demonstrates that our network with segregated dendritic compartments can engage in deep learning.

\subsection*{Coordinated local learning mimics backpropagation of error}
The backpropagation of error algorithm \citep{rumelhart_learning_1986} is still the primary learning algorithm used for deep supervised learning in artificial neural networks \citep{lecun_deep_2015}. Previous work has shown that learning with random feedback weights can actually match the synaptic weight updates specified by the backpropagation algorithm after a few epochs of training \citep{lillicrap_random_2014}. This fascinating observation suggests that deep learning with random feedback weights is not so much a different algorithm than backpropagation of error, but rather, networks with random feedback connections learn to approximate credit assignment as it is done in backpropagation \citep{lillicrap_random_2014}. Hence, we were curious as to whether or not our network was, in fact, learning to approximate the synaptic weight updates prescribed by backpropagation. To test this, we trained our one hidden layer network as before, but now, in addition to calculating the vector of hidden layer synaptic weight updates specified by our local learning rule ($\Delta \bm{W}^0$ in equation \eqref{eqn:updates}), we also calculated the vector of hidden layer synaptic weight updates that would be specified by non-locally backpropagating the error from the output layer, ($\Delta \bm{W}^0_{BP}$). We then calculated the angle between these two alternative weight updates. In a very high-dimensional space, any two independent vectors will be roughly orthogonal to each other (i.e. $\Delta \bm{W}^0 \angle \Delta \bm{W}^0_{BP} \approx  90^\circ$). If the two synaptic weight update vectors are \textit{not} orthogonal to each other (i.e. $\Delta \bm{W}^0 \angle \Delta \bm{W}^0_{BP} < 90^\circ$), then it suggests that the two algorithms are specifying similar weight updates.
\par
As in previous work \citep{lillicrap_random_2014}, we found that the initial weight updates for our network were orthogonal to the updates specified by backpropagation. But, as the network learned the angle dropped to approximately $65^\circ$, before rising again slightly to roughly $70^\circ$ (\hyperref[fig:F6]{Figure 6A}, blue line). This suggests that our network was learning to develop local weight updates in the hidden layer that were in rough agreement with the updates that explicit backpropagation would produce. However, this drop in orthogonality was still much less than that observed in non-spiking artificial neural networks learning with random feedback weights, which show a drop to below $45^\circ$ \citep{lillicrap_random_2014}. We suspected that the higher angle between the weight updates that we observed may have been because we were using spikes to communicate the feedback from the upper layer, which could introduce both noise and bias in the estimates of the output layer activity. To test this, we also examined the weight updates that our algorithm would produce if we propagated the spike rates of the output layer neurons, $\bm{\lambda}^U(t)$, back directly through the random feedback weights, $\bm{Y}$. In this scenario, we observed a much sharper drop in the $\Delta \bm{W}^0 \angle \Delta \bm{W}^0_{BP}$ angle, which reduced to roughly $35^\circ$  before rising again to $40^\circ$ (\hyperref[fig:F6]{Figure 6A}, red line). These results show that, in principle, our algorithm is learning to approximate the backpropagation algorithm, though with some drop in accuracy introduced by the use of spikes to propagate output layer activities to the hidden layer.
\par
To further examine how our local learning algorithm compared to backpropagation we compared the low-level features that the two algorithms learned. To do this, we trained the one hidden layer network with both our algorithm and backpropagation. We then examined the receptive fields (i.e. the synaptic weights) produced by both algorithms in the hidden layer synapses ($\bm{W}^0$) after 60 epochs of training. The two algorithms produced qualitatively similar receptive fields (\hyperref[fig:F6]{Figure 6C}). Both produced receptive fields with clear, high-contrast features for detecting particular strokes or shapes. To quantify the similarity, we conducted pair-wise correlation calculations for the receptive fields produced by the two algorithms and identified the maximum correlation pairs for each. Compared to shuffled versions of the receptive fields, there was a very high level of maximum correlation (\hyperref[fig:F6]{Figure 6B}), showing that the receptive fields were indeed quite similar. Thus, the data demonstrate that our learning algorithm using random feedback weights into segregated dendrites can in fact come to approximate the backpropagation of error algorithm.

\subsection*{Conditions on feedback weights}
Once we had convinced ourselves that our learning algorithm was, in fact, producing deep learning similar to that produced by backpropagation of error, we wanted to examine some of the constraints on learning. First, we wanted to explore the structure of the feedback weights. In our initial simulations we used non-sparse, random (i.e. normally distributed) feedback weights. We were interested in whether learning could still work with sparse weights, given that neocortical connectivity is sparse. As well, we wondered whether symmetric weights would \textit{improve} learning, which would be expected given previous findings \citep{lillicrap_random_2014,lee_target_2014,liao_how_2015}. To explore these questions, we trained our one hidden layer network using both sparse feedback weights (only 20\% non-zero values) and symmetric weights ($\bm{Y} = {\bm{W}^1}^T$) (\hyperref[fig:F7]{Figure 7A,C}). We found that learning actually \textit{improved} slightly with sparse weights (\hyperref[fig:F7]{Figure 7B}, red line), achieving an average error rate of 3.7\% by the 60\textsuperscript{th} epoch, compared to the average 4.1\% error rate achieved with fully random weights. But, this result appeared to depend on the magnitude of the sparse weights. To compensate for the loss of 80\% of the weights we initially increased the sparse synaptic weight magnitudes by a factor of 5. However, when we did not re-scale the sparse weights learning was actually \textit{worse} (\hyperref[fig:F7S1]{Figure 7, Supplement 1}). This suggests that sparse feedback provides a signal that is sufficient for credit assignment, but only if it is of appropriate magnitude.
\par
Similar to sparse feedback weights, symmetric feedback weights also improved learning, leading to a rapid decrease in the test error and an error rate of 3.6\% by the 60\textsuperscript{th} epoch (\hyperref[fig:F7]{Figure 7D}, red line). However, when we added noise to the symmetric weights this advantage was eliminated and learning was, in fact, slightly impaired (\hyperref[fig:F7]{Figure 7D}, blue line). At first, this was a very surprising result: given that learning works with random feedback weights, why would it not work with symmetric weights with noise? However, when we considered our previous finding that during the first epoch the feedforward weights, $\bm{W}^1$, learn to match the inverse of the feedback function (\hyperref[fig:F4S1]{Figure 4, Supplement 1}) a possible answer becomes clear. In the case of symmetric feedback weights the synaptic matrix $\bm{Y}$ is changing as $\bm{W}^1$ changes. This works fine when $\bm{Y}$ is set to $\bm{W}^{1^T}$, since that artificially forces weight alignment. But, if the feedback weights are set to $\bm{W}^{1^T}$ plus noise, then the system can never align the weights appropriately, since $\bm{Y}$ is now a moving target. This would imply that any implementation of feedback learning must either be very effective (to achieve the right feedback) or very slow (to allow the feedforward weights to adapt).

\subsection*{Learning with partial apical attenuation}
Another constraint that we wished to examine was whether total segregation of the apical inputs as we had done was necessary, given that real pyramidal neurons only show an attenuation of distal apical inputs to the soma \citep{larkum_new_1999}. To examine this, we re-ran our two hidden layer network, but now, we allowed the apical dendritic voltage to influence the somatic voltage by setting $g_A = 0.05$. This value gave us twelve times more attenuation than the attenuation from the basal compartments (since $g_B = 0.6$). This difference in the levels of attenuation is in-line with experimental data \citep{larkum_new_1999,larkum_synaptic_2009} (\hyperref[fig:F8]{Figure 8A}). When we compared the learning in this scenario to the scenario with total apical segregation, we observed very little difference in the error rates on the test set (\hyperref[fig:F8]{Figure 8B}, gray and red lines). Importantly, though, we found that if we decreased the apical attenuation to the same level as the basal compartment ($g_A = g_B = 0.6$) then the learning was significantly impaired (\hyperref[fig:F8]{Figure 8B}, blue line). This demonstrates that although total apical attenuation is not necessary, partial segregation of the apical compartment from the soma is necessary. This result makes sense given that our local targets for the hidden layer neurons incorporate a term that is supposed to reflect the response of the output neurons to the feedforward sensory information ($\bm{\alpha}^f$). Without some sort of separation of feedforward and feedback information, as is assumed in other models of deep learning \citep{lillicrap_random_2014,lee_target_2014}, this feedback signal would get corrupted by recurrent dynamics in the network. Our data show that electrontonically segregated dendrites is one potential way to achieve the required separation between feedforward and feedback information.

\section*{Discussion}
\label{discussion}
Deep learning has radically altered the field of AI, demonstrating that parallel distributed processing across multiple layers can produce human/animal-level capabilities in image classification, pattern recognition and reinforcement learning \citep{hinton_fast_2006,lecun_deep_2015,mnih_human-level_2015,silver_mastering_2016,krizhevsky_imagenet_2012,he_delving_2015}. Deep learning was motivated by analogies to the real brain \citep{lecun_deep_2015,cox_neural_2014}, so it is tantalizing that recent studies have shown that deep neural networks develop representations that strongly resemble the representations observed in the mammalian neocortex \citep{khaligh-razavi_deep_2014,yamins_using_2016,cadieu_deep_2014,kubilius_deep_2016}. In fact, deep learning models can match cortical representations even more so than some models that explicitly attempt to mimic the real brain \citep{khaligh-razavi_deep_2014}. Hence, at a phenomenological level, it appears that deep learning, defined as multilayer cost function reduction with appropriate credit assignment, may be key to the remarkable computational prowess of the mammalian brain \citep{marblestone2016towards}. However, the lack of biologically feasible mechanisms for credit assignment in deep learning algorithms, most notably backpropagation of error \citep{rumelhart_learning_1986}, has left neuroscientists with a mystery. How can the real brain solve the credit assignment problem (\hyperref[fig:F1]{Figure 1})? Here, we expanded on an idea that previous authors have explored \citep{kording_supervised_2001,spratling_cortical_2002,spratling_feedback_2006} and demonstrated that segregating the feedback and feedforward inputs to neurons, much as the real neocortex does \citep{larkum_new_1999,larkum_dendritic_2007,larkum_synaptic_2009}, can enable the construction of local targets to assign credit appropriately to hidden layer neurons (\hyperref[fig:F2]{Figure 2}). With this formulation, we showed that we could use segregated dendritic compartments to coordinate learning across layers (\hyperref[fig:F3]{Figure 3} and \hyperref[fig:F4]{Figure 4}). This enabled our network to take advantage of multiple layers to develop representations of hand-written digits in hidden layers that enabled better levels of classification accuracy on the MNIST dataset than could be achieved with a single layer (\hyperref[fig:F5]{Figure 5}). Furthermore, we found that our algorithm actually approximated the weight updates that would be prescribed by backpropagation, and produced similar low-level feature detectors (\hyperref[fig:F6]{Figure 6}). As well, we showed that our basic framework works with sparse feedback connections (\hyperref[fig:F7]{Figure 7}) and more realistic, partial apical attenuation (\hyperref[fig:F8]{Figure 8}). Therefore, our work demonstrates that deep learning is possible in a biologically feasible framework, provided that feedforward and feedback signals are sufficiently segregated in different dendrites.
\par
Perhaps the most biologically unrealistic component of backpropagation is the use of non-local ``weight-transport'' (\hyperref[fig:F1]{Figure 1B}) \citep{grossberg_competitive_1987}. Our model builds on recent neural networks research demonstrating that weight transport is not, in fact, required for credit assignment \citep{lillicrap_random_2014,liao_how_2015,lee_target_2014}. By demonstrating that the credit assignment problem is solvable using feedback to generate spatially local weight updates, these studies have provided a major breakthrough in our understanding of how deep learning could work in the real brain. However, this previous research involved an implicit assumption of separate feedforward and feedback pathways (\hyperref[fig:F2]{Figure 2A}). Although this is a possibility, there is currently no evidence for separate feedback pathways to guide learning in the neocortex. In our study, we obtained separate feedforward and feedback information using electrotonically segregated dendrites (\hyperref[fig:F2]{Figure 2B}), in analogy to neocortical pyramidal neurons \citep{larkum_new_1999}. This segregation allowed us to perform credit assignment using direct feedback pathways between layers (\hyperref[fig:F3]{Figures 3-5}). Moreover, we found that even with partial attenuation of the conductance from the apical dendrites to the soma our network could engage in deep learning (\hyperref[fig:F8]{Figure 8}). It should be recognized, though, that although our learning algorithm achieved deep learning with spatially local update rules, we had to assume some temporal non-locality. Our credit assignment system depended on the difference between the target and forward plateaus, $\bm{\alpha}^t- \bm{\alpha}^f$, which occurred at different times (roughly a 50-60 ms gap). Although this is not an ideal solution, this small degree of temporal non-locality for synaptic update rules is in-line with known biological mechanisms, such as slowly decaying calcium transients or synaptic tags \citep{redondo_making_2011}. Hence, our model exhibited deep learning using only local information contained within the cells.
\par
In this work we adopted a similar strategy to the one taken by \citeauthor{lee_target_2014}'s (2015) difference target propagation algorithm, wherein the feedback from higher layers is used to construct local activity targets at the hidden layers. One of the reasons that we adopted this strategy is that it is appealing to think that feedback from upper layers may not simply be providing a signal for plasticity, but also a modulatory signal to push the hidden layer neurons towards a ``better'' activity pattern in \textit{real-time}. This sort of top-down modulation could be used by the brain to improve sensory processing in different contexts and engage in inference \citep{bengio_towards_2015}. Indeed, framing cortico-cortical feedback as a mechanism to modulate incoming sensory activity is a more common way of viewing feedback signals in the neocortex \citep{larkum_cellular_2013,gilbert_top-down_2013,zhang_long-range_2014,fiser_experience-dependent_2016}. In light of this, it is interesting to note that distal apical inputs in somatosensory cortex can help animals perform sensory discrimination tasks \citep{takahashi_active_2016,manita_top-down_2015}. However, in our model, we did not actually implement a system that altered the hidden layer activity to make sensory calculations---we simply used the feedback signals to drive learning. In-line with this view of top-down feedback, two recent papers have found evidence that cortical feedback can indeed guide feedforward sensory plasticity \citep{thompson_cortical_2016,yamada_context-2017}. Yet, there is no reason that feedback signals cannot provide both top-down modulation and a signal for learning \citep{spratling_cortical_2002}. In this respect, a potential future advance on our model would be to implement a system wherein the feedback actively ``nudges'' the hidden layers towards appropriate activity patterns in order to guide learning while also shaping perception. This proposal is reminiscent of the approach taken in previous computational models \citep{urbanczik_learning_2014,spratling_feedback_2006,kording_supervised_2001}. Future research could study how top-down modulation and a signal for credit assignment can be combined in deep learning models.
\par
It is important to note that there are many aspects of our model that are not biologically realistic. These include (but are not limited to) synaptic inputs that are not conductance based, synaptic weights that can switch from positive to negative (and \textit{vice-versa}), and plateau potentials that are considered instantaneous and are not an accurate reflection of calcium spike dynamics in real pyramidal neurons \citep{larkum_new_1999,larkum_synaptic_2009}. However, the intention with this study was not to achieve total biological realism, but rather to demonstrate that the sort of dendritic segregation that neocortical pyramidal neurons exhibit can subserve credit assignment. Although we found that deep learning can still work with realistic levels of passive conductance from the apical compartment to the somatic compartment (\hyperref[fig:F8]{Figure 8B}, red line) \citep{larkum_new_1999}, we also found that learning was impaired if the basal and apical compartments had equal conductance to the soma (\hyperref[fig:F8]{Figure 8B}, blue line). This is interesting, because it suggests that the unique physiology of neocortical pyramidal neurons may in fact reflect nature's solution to deep learning. Perhaps the relegation of feedback from higher-order brain regions to the electrotonically distant apical dendrites \citep{budd_extrastriate_1998,spratling_cortical_2002}, and the presence of active calcium spike mechanisms in the apical shaft \citep{larkum_synaptic_2009}, are mechanisms for coordinating local synaptic weight updates. If this is correct, then the inhibitory interneurons that target apical dendrites and limit active communication to the soma \citep{murayama_dendritic_2009} may be used by the neocortex to control learning. Although this is speculative, it is worth noting that current evidence supports the idea that neuromodulatory inputs carrying temporally precise salience information \citep{hangya_central_2015} can shut off interneurons to disinhibit the distal apical dendrites \citep{pi_cortical_2013,karnani_opening_2016,pfeffer_inhibition_2013,brombas_activity-dependent_2014}, and presumably, promote apical communication to the soma. Recent work suggests that the specific patterns of interneuron inhibition on the apical dendrites are very spatially precise and differentially timed to motor behaviours \citep{munoz_layer-specific_2017}, which suggests that there may well be coordinated physiological mechanisms for determining when cortico-cortical feedback is transmitted to the soma. Future research should examine whether these inhibitory and neuromodulatory mechanisms do, in fact, open the window for plasticity in the basal dendrites of pyramidal neurons, as our model, and some recent experimental work \citep{bittner_conjunctive_2015}, would predict.
\par
An additional issue that should be recognized is that the error rates which our network achieved were by no means as low as can be achieved with artificial neural networks, nor at human levels of performance \citep{lecun1998gradient,li_very_2016}. As well, our algorithm was not able to take advantage of very deep structures (beyond two hidden layers, the error rate did not improve). In contrast, increasing the depth of networks trained with backpropagation can lead to performance improvements \citep{li_very_2016}. But, these observations do not mean that our network was not engaged in deep learning. First, it is interesting to note that although the backpropagation algorithm is several decades old \citep{rumelhart_learning_1986}, it was long considered to be useless for training networks with more than one or two hidden layers \citep{bengio_scaling_2007}. Indeed, it was only the use of layer-by-layer training that initially led to the realization that deeper networks can achieve excellent performance \citep{hinton_fast_2006}. Since then, both the use of very large datasets (with millions of examples), and additional modifications to the backpropagation algorithm, have been key to making backpropagation work well on deeper networks \citep{sutskever_importance_2013,lecun_deep_2015}. Future studies could examine how our algorithm could incorporate current techniques used in machine learning to work better on deeper architectures. Second, we stress that our network was not designed to match the state-of-the-art in machine learning, nor human capabilities. To test our basic hypothesis (and to run our leaky-integration and spiking simulations in a reasonable amount of time) we kept the network small, we stopped training before it reached its asymptote, and we did not implement any add-ons to the learning to improve the error rates, such as convolution and pooling layers, initialization tricks, mini-batch training, drop-out, momentum or RMSProp \citep{sutskever_importance_2013,tieleman_lecture_2012,srivastava_dropout:_2014}. Indeed, it would be quite surprising if a relatively vanilla, small network like ours could come close to matching current performance benchmarks in machine learning. Third, although our network was able to take advantage of multiple layers to improve the error rate, there may be a variety of reasons that ever increasing depth didn't improve performance significantly. For example, our use of direct connections from the output layer to the hidden layers may have impaired the network's ability to coordinate synaptic updates between \textit{hidden} layers. As well, given our finding that the use of spikes produced weight updates that were less well-aligned to backpropagation (\hyperref[fig:F6]{Figure 6A}) it is possible that deeper architectures require mechanisms to overcome the inherent noisiness of spikes.
\par
One aspect of our model that we did not develop was the potential for learning at the feedback synapses. Although we used random synaptic weights for feedback, we also demonstrated that our model actually learns to meet the mathematical conditions required for credit assignment, namely the feedforward weights come to approximate the inverse of the feedback (\hyperref[fig:F4S1]{Figure 4, Supplement 1}). This suggests that it would be beneficial to develop a synaptic weight update rule for the feedback synapses that made this aspect of the learning better. Indeed, \cite{lee_target_2014} implemented an ``inverse loss function'' for their feedback synapses which promoted the development of feedforward and feedback functions that were roughly inverses of each other, leading to the emergence of auto-encoder functions in their network. In light of this, it is interesting to note that there is evidence for unique, ``reverse'' spike-timing-dependent synaptic plasticity rules in the distal apical dendrites of pyramidal neurons \citep{sjostrom_cooperative_2006,letzkus_learning_2006}, which have been shown to produce symmetric feedback weights and auto-encoder functions in artificial spiking networks \citep{burbank_depression-biased_2012,burbank_mirrored_2015}. Thus, it is possible that early in development the neocortex actually learns cortico-cortical feedback connections that help it to assign credit for later learning. Our work suggests that any experimental evidence showing that feedback connections learn to approximate the inverse of feedforward connections could be considered as evidence for deep learning in the neocortex.
\par
In summary, deep learning has had a huge impact on AI, but, to date, its impact on neuroscience has been limited. Nonetheless, given a number of findings in neurophysiology and modeling \citep{yamins_using_2016}, there is growing interest in understanding how deep learning may actually be achieved by the real brain \citep{marblestone2016towards}. Our results show that by moving away from point neurons, and shifting towards multi-compartment neurons that segregate feedforward and feedback signals, the credit assignment problem can be solved and deep learning can be achieved. Perhaps the dendritic anatomy of neocortical pyramidal neurons is important for nature's own deep learning algorithm.

\section*{Materials and Methods}
\label{methods}
Code for the model can be obtained from a GitHub repository (\href{https://github.com/jordan-g/Segregated-Dendrite-Deep-Learning}{https://github.com/jordan-g/Segregated-Dendrite-Deep-Learning}) \citep{guergiuev_2017}. For notational simplicity, we describe our model in the case of a network with only one hidden layer. We describe how this is extended to a network with multiple layers at the end of this section. As well, at the end of this section in \hyperref[tab:T1]{Table 1} we provide a table listing the parameter values we used for all of the simulations presented in this paper.

\subsection*{Neuronal dynamics}
\par
The network described here consists of an input layer with $\ell$ neurons, a hidden layer with $m$ neurons, and an output layer with $n$ neurons. Neurons in the input layer are simple Poisson spiking neurons whose rate-of-fire is determined by the intensity of image pixels (ranging from $0$ - $\lambda_{max}$). Neurons in the hidden layer are modeled using three functional compartments---basal dendrites with voltages $\bm{B}(t) = [B_1(t), B_2(t), ..., B_m(t)]$, apical dendrites with voltages $\bm{A}(t) = [A_1(t), A_2(t), ..., A_m(t)]$, and somata with voltages $\bm{C}(t) = [C_1(t), C_2(t), ..., C_m(t)]$. Feedforward inputs from the input layer and feedback inputs from the output layer arrive at basal and apical synapses, respectively. At basal synapses, presynaptic spikes are translated into post-synaptic potentials $\bm{PSP^B}(t) = [PSP^B_1(t), PSP^B_2(t), ..., PSP^B_\ell(t)]$ given by:

\begin{align}
\label{eqn:psp}
PSP^B_j(t) &= \sum_{s \in X_j} \kappa (t - s)
\end{align}

where $X_j$ is the set of pre-synaptic spike times at synapses with presynaptic input neuron $j$. $\kappa(t) = (e^{-t/\tau_L} - e^{-t/\tau_s})\Theta(t)/(\tau_L - \tau_s)$ is the response kernel, where $\tau_s$ and $\tau_L$ are short and long time constants, and $\Theta$ is the Heaviside step function. The post-synaptic potentials at apical synapses, $\bm{PSP^A}(t) = [PSP^A_1(t), PSP^A_2(t), ..., PSP^A_n(t)]$, are modeled in the same manner. The basal and apical dendritic potentials for neuron $i$ are then given by weighted sums of the post-synaptic potentials at either its basal or apical synapses:

\begin{align}
\label{eqn:dend_voltages_hidden}
\begin{split}
B_i(t) &= \sum_{j=1}^\ell W^0_{ij} PSP^B_j(t) + b^0_i \\
A_i(t) &= \sum_{j=1}^n Y_{ij} PSP^A_j(t)
\end{split}
\end{align}

where $\bm{b}^0 = [b^0_1, b^0_2, ..., b^0_m]$ are bias terms, $\bm{W}^0$ is the $m \times \ell$ matrix of feedforward weights for neurons in the hidden layer, and $\bm{Y}$ is the $m \times n$ matrix of their feedback weights. The somatic voltage for neuron $i$ evolves with leak as:

\begin{align}
\label{eqn:soma_voltage_hidden_apical_conductance_duplicate}
\frac{dC_i(t)}{dt} &= -g_LC_i(t) + g_B(B_i(t) - C_i(t)) + g_A(A_i(t) - C_i(t))
\end{align}

where $g_L$ is the leak conductance, $g_B$ is the conductance from the basal dendrite to the soma, and $g_A$ is the conductance from the apical dendrite to the soma. Equation \eqref{eqn:soma_voltage_hidden_apical_conductance_duplicate} is identical to equation \eqref{eqn:soma_voltage_hidden_apical_conductance} in \nameref{results}. Note that for simplicity's sake we are assuming a resting potential of 0 V and a membrane capacitance of 1 F, but these values are not important for the results.
\par
The instantaneous firing rates of neurons in the hidden layer are given by $\bm{\lambda}^C(t) = [\lambda^C_1(t), \lambda^C_2(t), ..., \lambda^C_m(t)]$, where $\lambda^C_i(t)$ is the result of applying a nonlinearity, $\sigma(\cdot)$, to the somatic potential $C_i(t)$. We chose $\sigma(\cdot)$ to be a simple sigmoidal function, such that:

\begin{align}
\label{eqn:spike_rates_duplicate}
\lambda^C_i(t) &= \lambda_{max}\sigma(C_i(t)) = \lambda_{max}\frac{1}{1 + e^{-C_i(t)}}
\end{align}

\par
Here, $\lambda_{max}$ is the maximum possible rate-of-fire for the neurons, which we set to 200 Hz. Note that equation \eqref{eqn:spike_rates_duplicate} is identical to equation \eqref{eqn:spike_rates} in \nameref{results}. Spikes are then generated using Poisson processes with these firing rates. We note that although the maximum rate was 200 Hz, the neurons rarely achieved anything close to this rate, and the average rate of fire in the neurons during our simulations was 24 Hz, in-line with neocortical firing rates \citep{steriade_natural_2001}.
\par
Units in the output layer are modeled using only two compartments, dendrites with voltages $\bm{V}(t) = [V_1(t), V_2(t), ..., V_n(t)]$ and somata with voltages $\bm{U}(t) = [U_1(t), U_2(t), ..., U_n(t)]$. $V_i(t)$ is given by:

\begin{align}
\label{eqn:dend_voltage_output}
V_i(t) &= \sum_{j=1}^m W^1_{ij} PSP^V_j(t) + b^1_i
\end{align}

where $\bm{PSP}^V(t) = [PSP^V_1(t), PSP^V_2(t), ..., PSP^V_m(t)]$ are the post-synaptic potentials at synapses that receive feedforward input from the hidden layer, and are calculated in the manner described by equation \eqref{eqn:psp}. $U_i(t)$ evolves as:

\begin{align}
\label{eqn:soma_voltage_output}
\frac{dU_i(t)}{dt} &= -g_LU_i(t) + g_D(V_i(t) - U_i(t)) + I_i(t)
\end{align}

where $g_L$ is the leak conductance, $g_D$ is the conductance from the dendrite to the soma, and $\bm{I}(t) = [I_1(t), I_2(t), ..., I_n(t)]$ are somatic currents that can drive output neurons toward a desired somatic voltage. For neuron $i$, $I_i$ is given by:

\begin{align}
\label{eqn:somatic_current}
I_i(t) &= g_{E_i}(t)(E_E - U_i(t)) + g_{I_i}(t)(E_I - U_i(t))
\end{align}

where $\bm{g_E}(t) = [g_{E_1}(t), g_{E_2}(t), ..., g_{E_n}(t)]$ and $\bm{g_I}(t) = [g_{I_1}(t), g_{I_2}(t), ..., g_{I_n}(t)]$ are time-varying excitatory \& inhibitory nudging conductances, and $E_E$ and $E_I$ are the excitatory and inhibitory reversal potentials. In our simulations, we set $E_E=8$ V and $E_I=-8$ V. During the target phase only, we set $g_{I_i} = 1$ and $g_{E_i} = 0$ for all units $i$ whose output should be minimal, and $g_{E_i} = 1$ and $g_{I_i} = 0$ for the unit whose output should be maximal. In this way, all units other than the ``target'' unit are silenced, while the ``target'' unit receives a strong excitatory drive. In the forward phase, $\bm{I}(t)$ is set to $0$. The Poisson spike rates $\bm{\lambda}^U(t) = [\lambda^U_1(t), \lambda^U_2(t), ..., \lambda^U_n(t)]$ are calculated as in equation \eqref{eqn:spike_rates_duplicate}.

\subsection*{Plateau potentials}
At the end of the forward and target phases, we calculate plateau potentials $\bm{\alpha}^f = [ \alpha^f_1, \alpha^f_2, ..., \alpha^f_m ]$ and $\bm{\alpha}^t = [ \alpha^t_1, \alpha^t_2, ..., \alpha^t_m ]$ for apical dendrites of hidden layer neurons, where $\alpha^f_i$ and $\alpha^t_i$ are given by:

\begin{align}
\begin{split}
\label{eqn:plateau_potential_duplicate}
\alpha^f_i &= \sigma(\int_{t_1-t_p}^{t_1} A_i(t) dt) \\
\alpha^t_i &= \sigma(\int_{t_2-t_p}^{t_2} A_i(t) dt) \\
\end{split}
\end{align}

where $t_1$ and $t_2$ are the end times of the forward and target phases, respectively, and $t_p = t_{i+1} - (t_i + \Delta t_s)$ is the integration time for the plateau potential. $\Delta t_s = 30$ ms is the settling time for the voltages. Note that equation \eqref{eqn:plateau_potential_duplicate} is identical to equation \eqref{eqn:plateau_potential} in \nameref{results}. These plateau potentials are used by hidden layer neurons to update their basal weights.

\subsection*{Weight updates}
All feedforward synaptic weights are updated at the end of each target phase. Output layer units update their synaptic weights $\bm{W}^1$ in order to minimize the loss function $L^1$ given in equation \eqref{eqn:loss_functions}. All average voltages are calculated after a delay $\Delta t_s$ from the start of a phase, which allows for the network to reach a steady state before averaging begins. In practice this means that the average somatic voltage for output layer neuron $i$ in the forward phase, $\overline{U_i}^f$, has the property

\begin{align}
\label{eqn:avg_soma_voltage_output}
\overline{U_i}^f \approx k_D\overline{V_i}^f = k_D \big( \sum_{j=1}^m W^1_{ij}\overline{PSP^V_j}^f + b^1_i \big)
\end{align}

where $k_D=g_D/(g_L + g_D)$. Thus,

\begin{align}
\label{eqn:error_signal_output}
\begin{split}
\frac{\partial L^1}{\partial \bm{W}^1} &\approx -k_D\lambda_{max}(\hat{\bm{\lambda}}^U - \lambda_{max}\sigma(\overline{\bm{U}}^f))\sigma^\prime(\overline{\bm{U}}^f) \circ \overline{\bm{PSP^V}}^f \\
\frac{\partial L^1}{\partial \bm{b}^1} &\approx -k_D\lambda_{max}(\hat{\bm{\lambda}}^U - \lambda_{max}\sigma(\overline{\bm{U}}^f))\sigma^\prime(\overline{\bm{U}}^f)
\end{split}
\end{align}

The dendrites in the output layer use this approximation of the gradient in order to update their weights using gradient descent:

\begin{align}
\label{eqn:weight_update_output}
\begin{split}
\bm{W}^1 &\rightarrow \bm{W}^1 - \eta^1 P^1 \frac{\partial L^1}{\partial \bm{W}^1} \\
\bm{b}^1 &\rightarrow \bm{b}^1 - \eta^1 P^1 \frac{\partial L^1}{\partial \bm{b}^1}
\end{split}
\end{align}

where $\eta^1$ is a learning rate constant, and $P^1$ is a scaling factor used to normalize the scale of the rate-of-fire function.
\par
In the hidden layer, basal dendrites update their synaptic weights $\bm{W}^0$ by minimizing the loss function $L^0$ given in equation \eqref{eqn:loss_functions}. We define the target rates-of-fire $\hat{\bm{\lambda}}^C = [\hat{\lambda}^C_1, \hat{\lambda}^C_2, ..., \hat{\lambda}^C_m]$ such that

\begin{align}
\label{eqn:hidden_target_duplicate}
\hat{\lambda}^C_i = \overline{\lambda_i^C}^f + \alpha^t_i - \alpha^f_i
\end{align}

where $\bm{\alpha}^f = [\alpha^f_1, \alpha^f_2, ..., \alpha^f_m]$ and $\bm{\alpha}^t = [\alpha^t_1, \alpha^t_2, ..., \alpha^t_m]$ are forward and target phase plateau potentials given in equation \eqref{eqn:plateau_potential_duplicate}. Note that equation \eqref{eqn:hidden_target_duplicate} is identical to equation \eqref{eqn:hidden_target} in \nameref{results}. These hidden layer target firing rates are similar to the targets used in difference target propagation \citep{lee_target_2014}. Using the fact that $\overline{\lambda_i^C}^f \approx \lambda_{max}\sigma(\overline{C_i}^f)$, we can say that $\hat{\lambda}^C_i - \lambda_{max}\sigma(\overline{C_i}^f) \approx \alpha^t_i - \alpha^f_i$, and hence:

\begin{align}
\label{eqn:error_signal_hidden}
\begin{split}
\frac{\partial L^0}{\partial \bm{W}^0} &\approx -k_B(\bm{\alpha}^t -  \bm{\alpha}^f)\lambda_{max}\sigma^\prime(\overline{\bm{C}}^f) \circ \overline{\bm{PSP^B}}^f \\
\frac{\partial L^0}{\partial \bm{b}^0} &\approx -k_B(\bm{\alpha}^t -  \bm{\alpha}^f)\lambda_{max}\sigma^\prime(\overline{\bm{C}}^f)
\end{split}
\end{align}

where $k_B=g_B/(g_L + g_B + g_A)$. Basal weights are updated in order to descend this gradient:

\begin{align}
\label{eqn:weight_update_hidden}
\begin{split}
\bm{W}^0 &\rightarrow \bm{W}^0 - \eta^0 P^0 \frac{\partial L^0}{\partial \bm{W}^0} \\
\bm{b}^0 &\rightarrow \bm{b}^0 - \eta^0 P^0 \frac{\partial L^0}{\partial \bm{b}^0}
\end{split}
\end{align}

Importantly, this update rule is fully local for the hidden layer neurons. It consists essentially of three terms, (1) the difference in the plateau potentials for the target and forward phases ($\bm{\alpha}^t -  \bm{\alpha}^f$), (2) the derivative of the spike rate function ($\lambda_{max}\sigma^\prime(\overline{\bm{C}}^f)$), and (3) the postsynaptic potentials ($\overline{\bm{PSP^B}}^f$). All three of these terms are values that a real neuron could theoretically calculate using some combination of molecular synaptic tags, calcium currents, and back-propagating action potentials.

\subsection*{Multiple hidden layers}
In order to extend our algorithm to deeper networks with multiple hidden layers, our model incorporates direct synaptic connections from the output layer to each hidden layer. Thus, each hidden layer receives feedback from the output layer through its own separate set of fixed, random weights. For example, in a network with two hidden layers, both layers receive the feedback from the output layer at their apical dendrites through backward weights $\bm{Y}^0$ and $\bm{Y}^1$. The local targets at each layer are then given by:

\begin{align}
\label{eqn:3_layer_firing_rates}
\hat{\bm{\lambda}}^U &= \overline{\bm{\lambda}^U}^t \\
\hat{\bm{\lambda}}^{C^1} &= \overline{\bm{\lambda}^{C^1}}^t + {\bm{\alpha}^1}^t - {\bm{\alpha}^1}^f \\
\hat{\bm{\lambda}}^{C^0} &= \overline{\bm{\lambda}^{C^0}}^t + {\bm{\alpha}^0}^t - {\bm{\alpha}^0}^f
\end{align}

where the superscripts $^0$ and $^1$ denote the first and second hidden layers, respectively.

The local loss functions at each layer are:

\begin{align}
\label{eqn:3_layer_loss_functions}
\begin{split}
L^2 &= || \hat{\bm{\lambda}}^U - \lambda_{max}\sigma(\overline{\bm{U}}^f) ||^2_2 \\
L^1 &= || \hat{\bm{\lambda}}^{C^1} - \lambda_{max}\sigma(\overline{\bm{C}^1}^f) ||^2_2 \\
L^0 &= || \hat{\bm{\lambda}}^{C^0} - \lambda_{max}\sigma(\overline{\bm{C}^0}^f) ||^2_2
\end{split}
\end{align}

where $L^2$ is the loss at the output layer. The learning rules used by the hidden layers in this scenario are the same as in the case with one hidden layer.

\subsection*{Learning rate optimization}
For each of the three network sizes that we present in this paper, a grid search was performed in order to find good learning rates. We set the learning rate for each layer by stepping through the range $[0.1, 0.3]$ with a step size of $0.02$. For each combination of learning rates, a neural network was trained for one epoch on the 60, 000 training examples, after which the network was tested on 10,000 test images. The learning rates that gave the best performance on the test set after an epoch of training were used as a basis for a second grid search around these learning rates that used a smaller step size of $0.01$. From this, the learning rates that gave the best test performance after 20 epochs were chosen as our learning rates for that network size.
\par
In all of our simulations, we used a learning rate of 0.19 for a network with no hidden layers, learning rates of 0.21 (output and hidden) for a network with one hidden layer, and learning rates of 0.23 (hidden layers) and 0.12 (output layer) for a network with two hidden layers. All networks with one hidden layer had 500 hidden layer neurons, and all networks with two hidden layers had 500 neurons in the first hidden layer and 100 neurons in the second hidden layer.

\subsection*{Training paradigm}
For all simulations described in this paper, the neural networks were trained on classifying handwritten digits using the MNIST database of 28 pixel $\times$ 28 pixel images. Initial feedforward and feedback weights were chosen randomly from a uniform distribution over a range that was calculated to produce voltages in the dendrites between $-6$ - $12$ V.
\par
Prior to training, we tested a network's initial performance on a set of 10,000 test examples. This set of images was shuffled at the beginning of testing, and each example was shown to the network in sequence. Each input image was encoded into Poisson spiking activity of the 784 input neurons representing each pixel of the image. The firing rate of an input neuron was proportional to the brightness of the pixel that it represents (with spike rates between $[0$ - $\lambda_{max}]$. The spiking activity of each of the 784 input neurons was received by the neurons in the first hidden layer. For each test image, the network underwent only a forward phase. At the end of this phase, the network's classification of the input image was given by the neuron in the output layer with the greatest somatic potential (and therefore the greatest spike rate). The network's classification was compared to the target classification. After classifying all 10,000 testing examples, the network's classification error was given by the percentage of examples that it did not classify correctly.
\par
Following the initial test, training of the neural network was done in an on-line fashion. All 60,000 training images were randomly shuffled at the start of each training epoch. The network was then shown each training image in sequence, undergoing a forward phase ending with a plateau potential, and a target phase ending with another plateau potential. All feedforward weights were then updated at the end of the target phase. At the end of the epoch (after all 60,000 images were shown to the network), the network was again tested on the 10,000 test examples. The network was trained for up to 60 epochs.

\subsection*{Simulation details}
For each training example, a minimum length of 50 ms was used for each of the forward and target phases. The lengths of the forward and target training phases were determined by adding their minimum length to an extra length term, which was chosen randomly from a Wald distribution with a mean of 2 ms and scale factor of 1. During testing, a fixed length of 500 ms was used for the forward transmit phase. Average forward and target phase voltages were calculated after a settle duration of $\Delta t_s = 30$ ms from the start of the phase.
\par
For simulations with randomly sampled plateau potential times (\hyperref[fig:F5S1]{Figure 5, Supplement 1}), the time at which each neuron's plateau potential occurred was randomly sampled from a folded normal distribution ($\mu=0, \sigma^2=3$) that was truncated ($\text{max}=5$) such that plateau potentials occurred between 0 ms and 5 ms before the start of the next phase. In this scenario, the average apical voltage in the last 30 ms was averaged in the calculation of the plateau potential for a particular neuron.
\par
The time-step used for simulations was $dt = 1$ ms. At each time-step, the network's state was updated bottom-to-top beginning with the first hidden layer and ending with the output layer. For each layer, dendritic potentials were updated, followed by somatic potentials, and finally their spiking activity. \hyperref[tab:T1]{Table 1} lists the simulation parameters and the values that were used in the figures presented.
\par
All code was written using the Python programming language version 2.7 (RRID: SCR\_008394) with the NumPy (RRID: SCR\_008633) and SciPy (RRID: SCR\_008058) libraries. The code is open source and is freely available at \href{https://github.com/jordan-g/Segregated-Dendrite-Deep-Learning}{https://github.com/jordan-g/Segregated-Dendrite-Deep-Learning} \citep{guergiuev_2017}. The data used to train the network was from the Mixed National Institute of Standards and Technology (MNIST) database, which is a modification of the original database from the National Institute of Standards and Technology (RRID: SCR\_006440) \citep{lecun1998gradient}. The MNIST database can be found at \href{http://yann.lecun.com/exdb/mnist/}{http://yann.lecun.com/exdb/mnist/}. Some of the simulations were run on the SciNet High-Performance Computing platform \citep{loken_scinet_2010}.

\section*{Proofs}
\label{proofs}

\subsection*{Theorem for loss function coordination}
\label{thm1sec}
The targets that we selected for the hidden layer (see equation \eqref{eqn:hidden_target}) were based on the targets used in \cite{lee_target_2014}. The authors of that paper provided a proof showing that their hidden layer targets guaranteed that learning in one layer helped reduce the error in the next layer. However, there were a number of differences between our network and theirs, such as the use of spiking neurons, voltages, different compartments, etc. Here, we modify the original \cite{lee_target_2014} proof slightly to prove \hyperref[thm1]{Theorem 1} (which is the formal statement of equation \eqref{eqn:coordination_condition}).
\par
One important thing to note is that the theorem given here utilizes a target for the hidden layer that is slightly different than the one defined in equation \eqref{eqn:hidden_target}. However, the target defined in equation \eqref{eqn:hidden_target} is a numerical approximation of the target given in \hyperref[thm1]{Theorem 1}. After the proof of we describe exactly how these approximations relate to the targets given here.

\begin{thm}
\label{thm1}

Consider a neural network with one hidden layer and an output layer. Let $\hat{\bm{\lambda}}{^C}^* = \overline{\bm{\lambda}^C}^f + \sigma(\bm{Y}\overline{\bm{\lambda}^U}^t) - \sigma(\bm{Y}\lambda_{max}\sigma(E[\overline{\bm{U}}^f]))$ be the target firing rates for neurons in the hidden layer, where $\sigma(\cdot)$ is a differentiable function. Assume that $\overline{\bm{U}}^f \approx k_D\overline{\bm{V}}^f$. Let  $\hat{\bm{\lambda}}^U = \overline{\bm{\lambda}^U}^t$ be the target firing rates for the output layer. Also, for notational simplicity, let $\beta(\bm{x}) \equiv \lambda_{max}\sigma(k_D \bm{W}^1\bm{x})$ and $\gamma(\bm{x}) \equiv \sigma(Y\bm{x})$. Theorem 1 states that if $\hat{\bm{\lambda}}^U - \lambda_{max}\sigma(E[\overline{\bm{U}}^f])$ is sufficiently small, and the Jacobian matrices $J_{\beta}$ and $J_{\gamma}$ satisfy the condition that the largest eigenvalue of $(I - J_{\beta}J_{\gamma})^T(I - J_{\beta}J_{\gamma})$ is less than $1$, then

\begin{align*}
|| \hat{\bm{\lambda}}^U - \lambda_{max}\sigma(k_D \bm{W}^1 \hat{\bm{\lambda}}^{C^*}) ||_2^2 < || \hat{\bm{\lambda}}^U - \lambda_{max}\sigma(E[\overline{\bm{U}}^f]) ||_2^2
\end{align*}

We note again that the proof for this theorem is essentially a modification of the proof provided in \cite{lee_target_2014} that incorporates our \hyperref[lem1]{Lemma 1} to take into account the expected value of $\overline{\bm{PSP^V}}^f$, given that spikes in the network are generated with non-stationary Poisson processes.

\end{thm}

\begin{proof}

\begin{align*}
\hat{\bm{\lambda}}^U - \lambda_{max}\sigma(k_D \bm{W}^1 \hat{\bm{\lambda}}^{C^*}) &\equiv \hat{\bm{\lambda}}^U - \beta(\hat{\bm{\lambda}}^{C^*}) \\
&= \hat{\bm{\lambda}}^U - \beta( \overline{\bm{\lambda}^C}^f + \gamma ( \overline{\bm{\lambda}^U}^t ) - \gamma (\lambda_{max}\sigma(E[\overline{\bm{U}}^f]) ) )
\end{align*}

\noindent Lemma 1 shows that $\lambda_{max}\sigma(E[\overline{\bm{U}}^f]) = \lambda_{max}\sigma(E[k_D \bm{W}^1 \overline{\bm{PSP^V}}^f]) \approx \lambda_{max}\sigma(k_D \bm{W}^1 \overline{\bm{\lambda}^C}^f)$ given a sufficiently large averaging time window. Assume that $\lambda_{max}\sigma(E[\overline{\bm{U}}^f]) = \lambda_{max}\sigma(k_D \bm{W}^1 \overline{\bm{\lambda}^C}^f) \equiv \beta( \overline{\bm{\lambda}^C}^f)$. Then,

\begin{align*}
\hat{\bm{\lambda}}^U - \beta(\hat{\bm{\lambda}}^{C^*}) &= \hat{\bm{\lambda}}^U - \beta(\overline{\bm{\lambda}^C}^f + \gamma ( \overline{\bm{\lambda}^U}^t ) - \gamma ( \beta( \overline{\bm{\lambda}^C}^f) ) )
\end{align*}

\noindent Let $\bm{e} = \overline{\bm{\lambda}^U}^t - \beta(\overline{\bm{\lambda}^C}^f)$. Applying Taylor's theorem,

\begin{align*}
\hat{\bm{\lambda}}^U - \beta(\hat{\bm{\lambda}}^{C^*}) &= \hat{\bm{\lambda}}^U - \beta(\overline{\bm{\lambda}^C}^f + J_{\gamma}\bm{e} + \bm{o}(|| \bm{e} ||_2) )
\end{align*}

\noindent where $\bm{o}(|| \bm{e} ||_2)$ is the remainder term that satisfies $\lim_{\bm{e}\to 0} \bm{o}(|| \bm{e} ||_2)/|| \bm{e} ||_2 = 0$. Applying Taylor's theorem again,

\begin{align*}
\hat{\bm{\lambda}}^U - \beta(\hat{\bm{\lambda}}^{C^*}) &= \hat{\bm{\lambda}}^U - \beta(\overline{\bm{\lambda}^C}^f) - J_{\beta}(J_{\gamma}\bm{e} + \bm{o}(|| \bm{e} ||_2)) \\
&\quad - \bm{o}(|| (J_{\gamma}\bm{e} + \bm{o}(|| \bm{e} ||_2) ||_2) \\
&= \hat{\bm{\lambda}}^U - \beta(\overline{\bm{\lambda}^C}^f) + J_{\beta}J_{\gamma}\bm{e} - \bm{o}(|| \bm{e} ||_2) \\
&= (I - J_{\beta}J_{\gamma})\bm{e} - \bm{o}(|| \bm{e} ||_2)
\end{align*}

\noindent Then,

\begin{align*}
|| \hat{\bm{\lambda}}^U - \beta(\hat{\bm{\lambda}}^{C^*}) ||_2^2 &= ((I - J_{\beta}J_{\gamma})\bm{e} - \bm{o}(|| \bm{e} ||_2))^T ((I - J_{\beta}J_{\gamma})\bm{e} - \bm{o}(|| \bm{e} ||_2)) \\
&= \bm{e}^T(I - J_{\beta}J_{\gamma})^T(I - J_{\beta}J_{\gamma})\bm{e} - \bm{o}(|| \bm{e} ||_2)^T(I - J_{\beta}J_{\gamma})\bm{e} \\
&\quad - \bm{e}^T(I - J_{\beta}J_{\gamma})^T\bm{o}(|| \bm{e} ||_2) + \bm{o}(|| \bm{e} ||_2)^T\bm{o}(|| \bm{e} ||_2) \\
&= \bm{e}^T(I - J_{\beta}J_{\gamma})^T(I - J_{\beta}J_{\gamma})\bm{e} + \bm{o}(|| \bm{e} ||_2^2) \\
&\leq \mu || \bm{e} ||_2^2 + | o(|| \bm{e} ||_2^2) |
\end{align*}

\noindent where $\mu$ is the largest eigenvalue of $(I - J_{\beta}J_{\gamma})^T(I - J_{\beta}J_{\gamma})$. If $\bm{e}$ is sufficiently small so that ${| o(|| \bm{e} ||_2^2)) | < (1 - \mu)|| \bm{e} ||_2^2}$, then

\begin{align*}
|| \hat{\bm{\lambda}}^U - \lambda_{max}\sigma(k_D \bm{W}^1 \hat{\bm{\lambda}}^{C^*}) ||_2^2 &\leq || \bm{e} ||_2^2 = || \hat{\bm{\lambda}}^U - \lambda_{max}\sigma(E[\overline{\bm{U}}^f]) ||_2^2
\end{align*}

\end{proof}

Note that the last step requires that $\mu$, the largest eigenvalue of $(I - J_{\beta}J_{\gamma})^T(I - J_{\beta}J_{\gamma})$, is below 1. Clearly, we do not actually have any guarantee of meeting this condition. However, our results show that even though the feedback weights are random and fixed, the feedforward weights actually learn to meet this condition during the first epoch of training (\hyperref[fig:F4S1]{Figure 4, Supplement 1}).

\subsection*{Hidden layer targets}
\label{hidden_targets}

\hyperref[thm1]{Theorem 1} shows that if we use a target $\hat{\bm{\lambda}}^{C^*} = \overline{\bm{\lambda}^C}^f + \sigma(\bm{Y} \overline{\bm{\lambda}^U}^t) - \sigma(\bm{Y} \lambda_{max}\sigma(k_D\bm{W}^1 \overline{\bm{\lambda}^C}^f))$ for the hidden layer, there is a guarantee that the hidden layer approaching this target will also push the upper layer closer to its target $\hat{\bm{\lambda}}^U$, if certain other conditions are met.
Our specific choice of $\hat{\bm{\lambda}}^C$ defined in the \nameref{results} (equation \eqref{eqn:hidden_target}) approximates this target rate vector using variables that are accessible to the hidden layer units.
\par
Because neuronal units calculate averages after the network has reached a steady state, $\lambda_{max}\sigma(\overline{\bm{U}}^f) \approx \overline{\bm{\lambda}^U}^f$ and $\overline{\bm{A}}^f \approx \bm{Y}\overline{\bm{\lambda}^U}^f$. Using \hyperref[lem1]{Lemma 1} and equation \eqref{eqn:average_vectors}, $E[\overline{\bm{U}}^f] \approx k_D\bm{W}^1 \overline{\bm{\lambda}^C}^f$. If we assume that $\overline{\bm{U}}^f \approx E[\overline{\bm{U}}^f]$, which is true on average, then:

\begin{align}
\label{eqn:alpha_f_approx}
\bm{\alpha}^f = \sigma(\overline{\bm{A}}^f) \approx \sigma(\bm{Y}\overline{\bm{\lambda}^U}^f) \approx \sigma(\bm{Y}\lambda_{max}\sigma(\overline{\bm{U}}^f)) \approx \sigma(\bm{Y}\lambda_{max}\sigma(k_D\bm{W}^1 \overline{\bm{\lambda}^C}^f))
\end{align}

and:

\begin{align}
\label{eqn:alpha_t_approx}
\bm{\alpha}^t = \sigma(\overline{\bm{A}}^t) \approx \sigma(\bm{Y}\overline{\bm{\lambda}^U}^t)
\end{align}

Therefore, $\hat{\bm{\lambda}}^C \approx \hat{\bm{\lambda}}^{C^*}$.

\par
Thus, our hidden layer targets ensure that our model employs a learning rule similar to difference target propagation that approximates the necessary conditions to guarantee error convergence.

\subsection*{Lemma for firing rates}
\label{lem1sec}
\hyperref[thm1]{Theorem 1} had to rely on the equivalence between the average spike rates of the neurons and the postsynaptic potentials. Here, we prove a lemma showing that this equivalence does indeed hold as long as the integration time is long enough relative to the synaptic time constants $t_s$ and $t_L$.

\begin{lem}
\label{lem1}

Let $X$ be a set of presynaptic spike times during the time interval $\Delta t = t_1 - t_0$, distributed according to an inhomogeneous Poisson process. Let $N = |X|$ denote the number of presynaptic spikes during this time window, and let $s_i \in X$ denote the $i$\textsuperscript{th} presynaptic spike, where $0 < i \leq N$. Finally, let $\lambda(t)$ denote the time-varying presynaptic firing rate (i.e. the time-varying mean of the Poisson process), and $PSP(t)$ be the postsynaptic potential at time $t$ given by equation \eqref{eqn:psp}. Then, during the time window $\Delta t$, as long as $\Delta t \gg 2\tau_L^2\tau_s^2\overline{\lambda^2} / (\tau_L - \tau_s)^2(\tau_L + \tau_s)$,

\begin{align*}
E [ \overline{PSP(t)} ] \approx \overline{\lambda}
\end{align*}

\end{lem}

\begin{proof}

\noindent The average of $PSP(t)$ over the time window $\Delta t$ is

\begin{align*}
\overline{PSP} &= \frac{1}{\Delta t} \int_{t_0}^{t_1} PSP(t) dt \\
               &= \frac{1}{\Delta t} \sum_{s \in X} \int_{t_0}^{t_1} \frac{e^{-(t-s)/\tau_L} - e^{-(t-s)/\tau_s}}{\tau_L - \tau_s} \Theta(t-s) dt
\end{align*}

\noindent Since $\Theta(t-s) = 0$ for all $t < s$,

\begin{align*}
\overline{PSP} &= \frac{1}{\Delta t} \sum_{s \in X} \int_{s}^{t_1} \frac{e^{-(t-s)/\tau_L} - e^{-(t-s)/\tau_s}}{\tau_L - \tau_s} dt \\
                  &= \frac{1}{\Delta t} \Bigg( N - \sum_{s \in X} \frac{\tau_L e^{-(t_1-s)/\tau_L} - \tau_s e^{-(t_1-s)/\tau_s}}{\tau_L - \tau_s} \Bigg)
\end{align*}

\noindent The expected value of $\overline{PSP}$ with respect to $X$ is given by

\begin{align*}
E_X[ \overline{PSP} ] &= E_X \Bigg[ \frac{1}{\Delta t} \Bigg( N - \sum_{s \in X} \frac{\tau_L e^{-(t_1-s)/\tau_L} - \tau_s e^{-(t_1-s)/\tau_s}}{\tau_L - \tau_s} \Bigg) \Bigg] \\
           &= \frac{E_X [ N ]}{\Delta t} - \frac{1}{\Delta t} E_X \Bigg[ \sum_{i = 0}^N \Bigg( \frac{\tau_L e^{-(t_1-s_i)/\tau_L} - \tau_s e^{-(t_1-s_i)/\tau_s}}{\tau_L - \tau_s} \Bigg) \Bigg]
\end{align*}

\noindent Since the presynaptic spikes are an inhomogeneous Poisson process with a rate $\lambda$, $E_X [ N ] = \int_{t_0}^{t_1} \lambda dt$. Thus,

\begin{align*}
E_X[ \overline{PSP} ] &= \frac{1}{\Delta t} \int_{t_0}^{t_1} \lambda dt - \frac{1}{\Delta t} E_X \Bigg[ \sum_{i = 0}^N g(s_i) \Bigg] \\
&= \overline{\lambda} - \frac{1}{\Delta t} E_X \Bigg[ \sum_{i = 0}^N  g(s_i) \Bigg]
\end{align*}

\noindent where we let $g(s_i) \equiv (\tau_L e^{-(t_1-s_i)/\tau_L} - \tau_s e^{-(t_1-s_i)/\tau_s})/(\tau_L - \tau_s)$. Then, the law of total expectation gives

\begin{align*}
E_X \Bigg[ \sum_{i = 0}^N g(s_i) \Bigg] &= E_N \Bigg[ E_X \Bigg[ \sum_{i = 0}^N g(s_i) \Bigg| N \Bigg] \Bigg] \\
                                      &= \sum_{n=0}^{\infty} \Bigg( E_X \Bigg[ \sum_{i = 0}^N g(s_i) \Bigg| N = n \Bigg] \cdot P(N = n) \Bigg)
\end{align*}

\noindent Letting $f_{s_i}(s)$ denote $P(s_i=s)$, we have that

\begin{align*}
E_X \Bigg[ \sum_{i = 0}^N g(s_i) \Bigg| N = n \Bigg] &= \sum_{i = 0}^n E_X [ g(s_i) ] \\
                                                   &= \sum_{i = 0}^n \int_{t_0}^{t_1} g(s)f_{s_i}(s) ds
\end{align*}

\noindent Since Poisson spike times are independent, for an inhomogeneous Poisson process:

\begin{align*}
f_{s_i}(s) &= \frac{\lambda(s)}{\int_{t_0}^{t_1} \lambda(t) dt} \\
&= \frac{\lambda(s)} {\overline{\lambda}\Delta t}
\end{align*}

for all $s \in [t_0, t_1]$. Since Poisson spike times are independent, this is true for all $i$. Thus,

\begin{align*}
E_X \Bigg[ \sum_{i = 0}^N g(s_i) \Bigg| N = n \Bigg] &= \frac{1}{\overline{\lambda} \Delta t} \sum_{i = 0}^n \int_{t_0}^{t_1} g(s) \lambda(s) ds \\
                                                   &= \frac{n}{\overline{\lambda} \Delta t} \int_{t_0}^{t_1} g(s) \lambda(s) ds \\
\end{align*}

\noindent Then,

\begin{align*}
E_X \Bigg[ \sum_{i = 0}^N g(s_i) \Bigg] &= \sum_{n=0}^{\infty} \Bigg( \frac{n}{\overline{\lambda} \Delta t} \Bigg( \int_{t_0}^{t_1} g(s) \lambda(s) ds \Bigg) \cdot P(N = n) \Bigg) \\
                                      &= \frac{1}{\overline{\lambda} \Delta t} \Bigg( \int_{t_0}^{t_1} g(s) \lambda(s) ds \Bigg) \Bigg( \sum_{n=0}^{\infty} n \cdot P(N = n) \Bigg)
\end{align*}

\noindent Now, for an inhomogeneous Poisson process with time-varying rate $\lambda(t)$,

\begin{align*}
P(N=n) &= \frac{[\int_{t_0}^{t_1} \lambda(t) dt]^n e^{-\int_{t_0}^{t_1} \lambda(t) dt}}{n!} \\
&= \frac{[\overline{\lambda} \Delta t]^n e^{-(\overline{\lambda} \Delta t)}}{n!}
\end{align*}

\noindent Thus,

\begin{align*}
E_X \Bigg[ \sum_{i = 0}^N g(s_i) \Bigg] &= \frac{e^{-(\overline{\lambda} \Delta t)}}{\overline{\lambda} \Delta t} \Bigg( \int_{t_0}^{t_1} g(s) \lambda(s) ds \Bigg) \Bigg( \sum_{n=0}^{\infty} n \frac{[\overline{\lambda} \Delta t]^n}{n!} \Bigg) \\
                                      &= \frac{e^{-(\overline{\lambda} \Delta t)}}{\overline{\lambda} \Delta t} \Bigg( \int_{t_0}^{t_1} g(s) \lambda(s) ds \Bigg) (\overline{\lambda} \Delta t) e^{\overline{\lambda} \Delta t} \\
                                      &= \int_{t_0}^{t_1} g(s) \lambda(s) ds
\end{align*}

\noindent Then,

\begin{align*}
E_X[ \overline{PSP} ] &= \overline{\lambda} - \frac{1}{\Delta t} \Bigg( \int_{t_0}^{t_1} g(s) \lambda(s) ds \Bigg)
\end{align*}

\noindent The second term of this equation is always greater than or equal to 0, since $g(s) \geq 0$ and $\lambda(s) \geq 0$ for all $s$. Thus, $E_X[ \overline{PSP} ] \leq \overline{\lambda}$. As well, the Cauchy-Schwarz inequality states that

\begin{align*}
\int_{t_0}^{t_1} g(s) \lambda(s) ds &\leq \sqrt{\int_{t_0}^{t_1} g(s)^2 ds} \sqrt{\int_{t_0}^{t_1} \lambda(s)^2 ds} \\
&= \sqrt{\int_{t_0}^{t_1} g(s)^2 ds} \sqrt{\overline{\lambda^2}\Delta t}
\end{align*}

\noindent where

\begin{align*}
\int_{t_0}^{t_1} g(s)^2 ds &= \int_{t_0}^{t_1} \Big( \frac{\tau_L e^{-(t_1-s)/\tau_L} - \tau_s e^{-(t_1-s)/\tau_s}}{\tau_L - \tau_s} \Big)^2 ds \\
&\leq \frac{1}{2(\tau_L - \tau_s)^2} \Bigg( 4\frac{\tau_L^2\tau_s^2}{\tau_L + \tau_s} \Bigg) \\
&= \frac{2\tau_L^2\tau_s^2}{(\tau_L - \tau_s)^2(\tau_L + \tau_s)}
\end{align*}

\noindent Thus,

\begin{align*}
\int_{t_0}^{t_1} g(s) \lambda(s) ds &\leq \sqrt{\frac{2\tau_L^2\tau_s^2}{(\tau_L - \tau_s)^2(\tau_L + \tau_s)}} \sqrt{\overline{\lambda^2}\Delta t} \\
&= \sqrt{\Delta t}\sqrt{\frac{2\tau_L^2\tau_s^2\overline{\lambda^2}}{(\tau_L - \tau_s)^2(\tau_L + \tau_s)}}
\end{align*}

\noindent Therefore,

\begin{align*}
E_X[ \overline{PSP} ] &\geq \overline{\lambda} - \frac{1}{\Delta t}\sqrt{\Delta t}\sqrt{\frac{2\tau_L^2\tau_s^2\overline{\lambda^2}}{(\tau_L - \tau_s)^2(\tau_L + \tau_s)}} \\
&= \overline{\lambda} - \sqrt{\frac{2\tau_L^2\tau_s^2\overline{\lambda^2}}{\Delta t(\tau_L - \tau_s)^2(\tau_L + \tau_s)}}
\end{align*}

\noindent Then,

\begin{align*}
\overline{\lambda} - \sqrt{\frac{2\tau_L^2\tau_s^2\overline{\lambda^2}}{\Delta t(\tau_L - \tau_s)^2(\tau_L + \tau_s)}} \leq E_X[ \overline{PSP} ] \leq \overline{\lambda}
\end{align*}

\noindent Thus, as long as $\Delta t \gg 2\tau_L^2\tau_s^2\overline{\lambda^2} / (\tau_L - \tau_s)^2(\tau_L + \tau_s)$, $E_X[ \overline{PSP} ] \approx \overline{\lambda}$.
\end{proof}

What this lemma says, effectively, is that the expected value of $PSP$ is going to be roughly the average presynaptic rate of fire as long as the time over which the average is taken is sufficiently long in comparison to the postsynaptic time constants and the average rate-of-fire is sufficiently small. In our simulations, $\Delta t$ is always greater than or equal to 50 ms, the average rate-of-fire is approximately 20 Hz, and our time constants $\tau_L$ and $\tau_s$ are 10 ms and 3 ms, respectively. Hence, in general:

\begin{align*}
2\tau_L^2\tau_s^2\overline{\lambda^2} / (\tau_L - \tau_s)^2(\tau_L + \tau_s) &= 2 (10)^2(3)^2(0.02)^2 / (10 - 3)^2(10 + 3) \\
&\approx 0.001 \\
&\ll 50
\end{align*}

Thus, in the proof of \hyperref[thm1]{Theorem 1}, we assume $E_X[\overline{PSP}]=\overline{\lambda}$.
\clearpage

\begin{table}
{\renewcommand{\arraystretch}{1.5}%
\begin{tabular}{ | c | c | c | p{9cm} |}
\hline
\textbf{Parameter} & \textbf{Units} & \textbf{Value} & \textbf{Description} \\ \hline
$dt$ & ms & 1 & Time step resolution \\ \hline
$\lambda_{max}$ & Hz & 200 & Maximum spike rate \\ \hline
$\tau_s$ & ms & 3 & Short synaptic time constant \\ \hline
$\tau_L$ & ms & 10 & Long synaptic time constant \\ \hline
$\Delta t_s$ & ms & 30 & Settle duration for calculation of average voltages \\ \hline
$g_B$ & S & 0.6 & Hidden layer conductance from basal dendrite to the soma \\ \hline
$g_A$ & S & 0, 0.05, 0.6 & Hidden layer conductance from apical dendrite to the soma \\ \hline
$g_D$ & S & 0.6 & Output layer conductance from dendrite to the soma \\ \hline
$g_L$ & S & 0.1 & Leak conductance \\ \hline
$P_0$ & -- & $20/\lambda_{max}$ & Hidden layer error signal scaling factor \\ \hline
$P_1$ & -- & $20/\lambda_{max}^2$ & Output layer error signal scaling factor \\ \hline
\end{tabular}}\\
\caption{List of parameter values used in our simulations.}
\label{tab:T1}
\end{table}
\clearpage

\section*{Acknowledgments}
We would like to thank Douglas Tweed, Jo\~ao Sacramento, Walter Senn, and Yoshua Bengio for helpful discussions on this work. This research was supported by two grants to B.A.R.: a Discovery Grant from the Natural Sciences and Engineering Research Council of Canada (RGPIN-2014-04947) and a Google Faculty Research Award. The authors declare no competing financial interests. Some simulations were performed on the gpc supercomputer at the SciNet HPC Consortium. SciNet is funded by: the Canada Foundation for Innovation under the auspices of Compute Canada; the Government of Ontario; Ontario Research Fund - Research Excellence; and the University of Toronto.

\bibliographystyle{apa}
\bibliography{library}

\clearpage

\section*{Figures}
\label{figures}

\begin{figure}[!htb]
{\centering
\includegraphics[width=\textwidth]{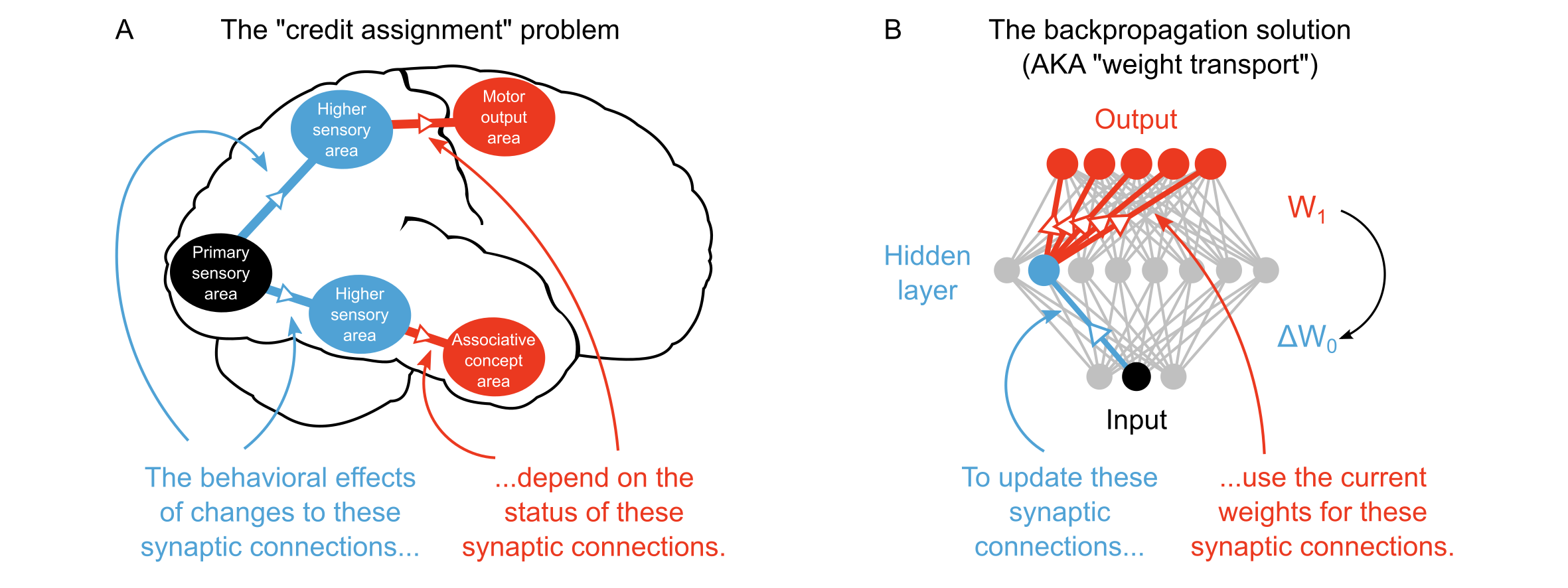}
\caption{\textbf{The credit assignment problem in multi-layer neural networks.}}
\par}
\medskip
\small  
(\textbf{A}) Illustration of the credit assignment problem. In order to take full advantage of the multi-circuit architecture of the neocortex when learning, synapses in earlier processing stages (blue connections) must somehow receive ``credit'' for their impact on behavior or cognition. However, the credit due to any given synapse early in a processing pathway depends on the downstream synaptic connections that link the early pathway to later computations (red connections). (\textbf{B}) Illustration of weight transport in backpropagation. To solve the credit assignment problem, the backpropagation of error algorithm explicitly calculates the credit due to each synapse in the hidden layer by using the downstream synaptic weights when calculating the hidden layer weight changes. This solution works well in AI applications, but is unlikely to occur in the real brain.
\label{fig:F1}
\end{figure}
\clearpage

\begin{figure}[!htb]
{\centering
\includegraphics[width=\textwidth]{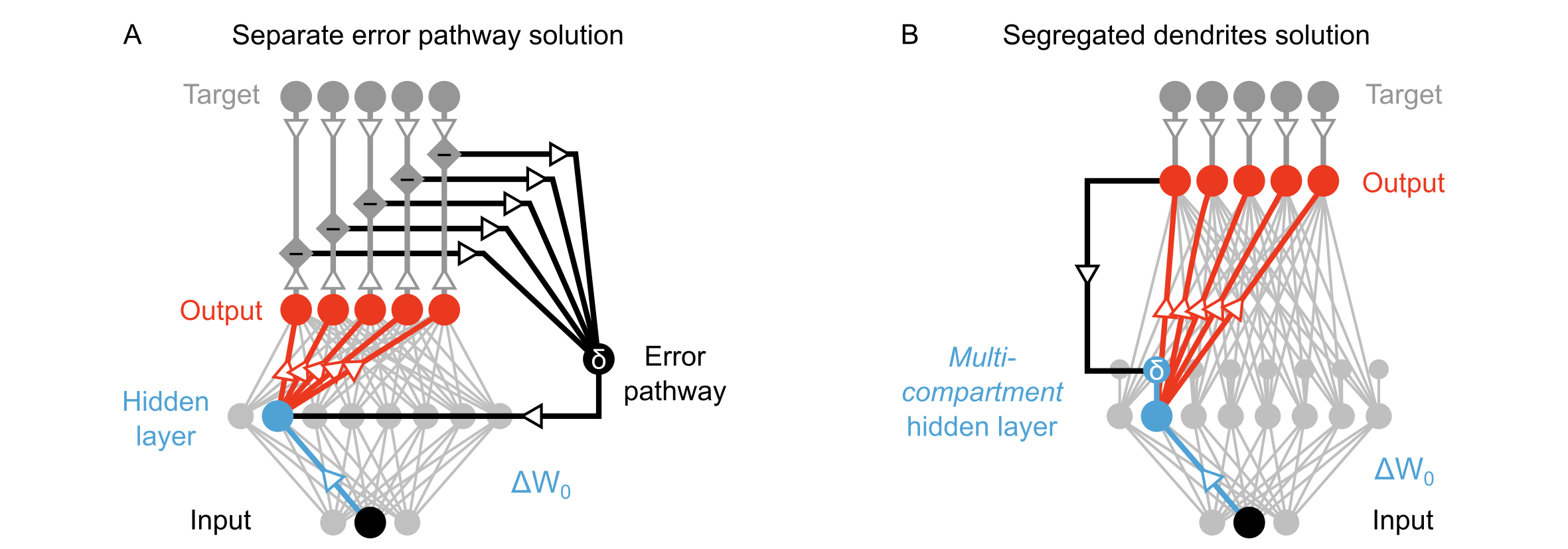}
\caption{\textbf{Potential solutions to credit assignment using top-down feedback.}}
\par}
\medskip
\small  
(\textbf{A}) Illustration of the implicit feedback pathway used in previous models of deep learning. In order to assign credit, feedforward information must be integrated separately from any feedback signals used to calculate error for synaptic updates (the error is indicated here with $\delta$). (\textbf{B}) Illustration of the segregated dendrites proposal. Rather than using a separate pathway to calculate error based on feedback, segregated dendritic compartments could receive feedback and calculate the error signals locally.
\label{fig:F2}
\end{figure}
\clearpage

\begin{figure}[!htb]
{\centering
\includegraphics[width=\textwidth]{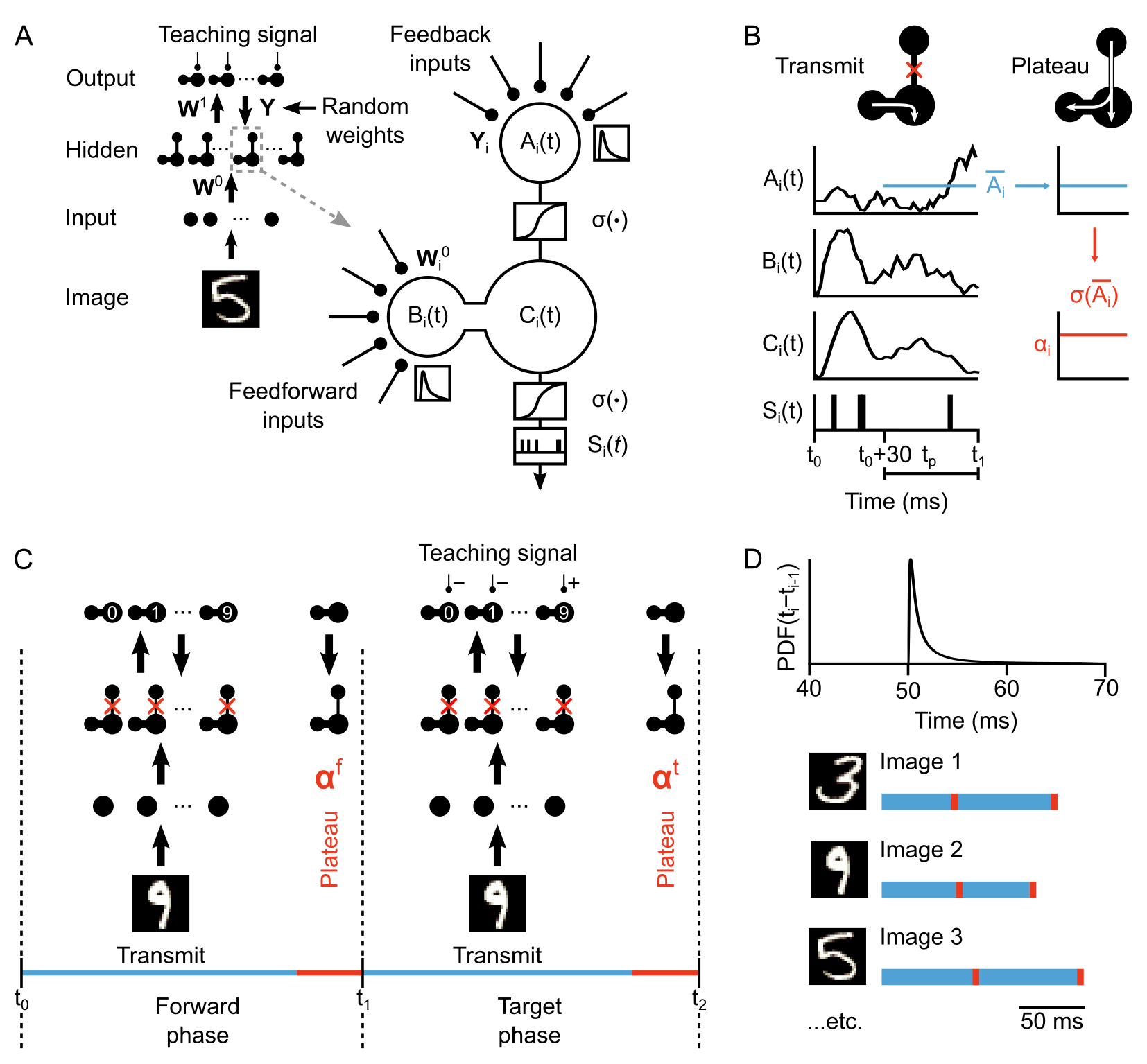}
\caption{\textbf{Illustration of a multi-compartment neural network model for deep learning.}}
\par}
\medskip
\small
(\textbf{A}) \textit{Left}: Diagram of our network architecture. An input image is represented by the spiking of input units which propagate to the hidden layer through weights $\bm{W}^0$. Hidden layer activities arrive at the output layer through weights $\bm{W}^1$. Feedback from the output layer is sent back to the hidden layer through fixed, random weights $\bm{Y}$. \textit{Right}: Illustration of unit $i$ in the hidden layer. The unit has a basal dendrite, soma, and apical dendrite with membrane potentials $B_i$, $C_i$ and $A_i$, respectively. The apical dendrite communicates to the soma predominantly using non-linear signals $\sigma(\cdot)$. The spiking output of the cell, $S_i(t)$, is a Poisson process with rate $\lambda_{max}\sigma(C_i(t))$.
(\textbf{B}) Illustration of neuronal dynamics in the two modes of processing. In the transmit mode, the apical dendrite accumulates a measure of its average membrane potential from $t_1 - t_p$ to $t_1$. In the plateau mode, an apical plateau potential equal to the nonlinearity $\sigma(\cdot)$ applied to the average apical potential travels to the soma and basal dendrite.
(\textbf{C}) Illustration of the sequence of network phases that occur for each training example. The network undergoes a forward phase (transmit \& plateau) and a target phase (transmit \& plateau) in sequence. In the target phase, a teaching signal is introduced to the output layer, shaping its activity.
(\textbf{D}) Illustration of phase length sampling. For each training example, the lengths of forward \& target phases are randomly drawn from a shifted inverse Gaussian distribution with a minimum of 50 ms.
\label{fig:F3}
\end{figure}
\clearpage

\begin{figure}[!htb]
{\centering
\includegraphics[width=\textwidth]{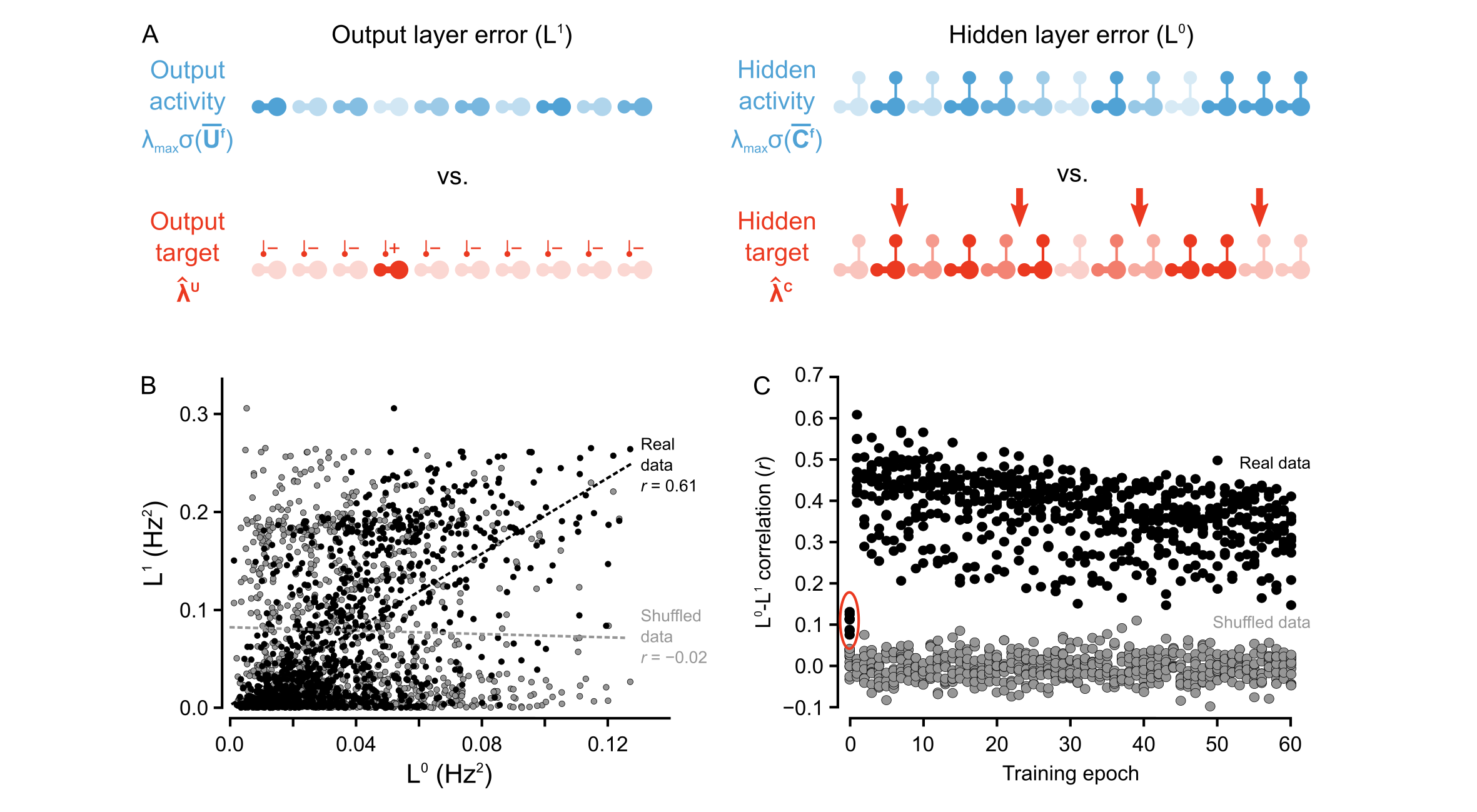}
\caption{\textbf{Co-ordinated errors between the output and hidden layers.}}
\par}
\medskip
\small
(\textbf{A}) Illustration of output error ($L^1$) and local hidden error ($L^0$). For a given test example shown to the network in a forward phase, the output layer error is defined as the squared norm of the difference between target firing rates $\bm{\hat{\lambda}}^U$ and the sigmoid function applied to the average somatic potentials $\bm{\overline{U}}^f$ across all output units. Hidden layer error is defined similarly, except the target is $\bm{\hat{\lambda}}^C$ (as defined in the text).
(\textbf{B}) Plot of $L^1$ vs. $L^0$ for all of the `2' images after one epoch of training. There is a strong correlation between hidden layer error and output layer error (real data, black), as opposed to when output and hidden errors were randomly paired (shuffled data, gray).
(\textbf{C}) Plot of correlation between hidden layer error and output layer error across training for each category of images (each dot represents one category). The correlation is significantly higher in the real data than the shuffled data throughout training. Note also that the correlation is much lower on the first epoch of training (red oval), suggesting that the conditions for credit assignment are still developing during the first epoch.
\label{fig:F4}
\end{figure}
\clearpage

\begin{figure}[!htb]
{\centering
\includegraphics[width=\textwidth]{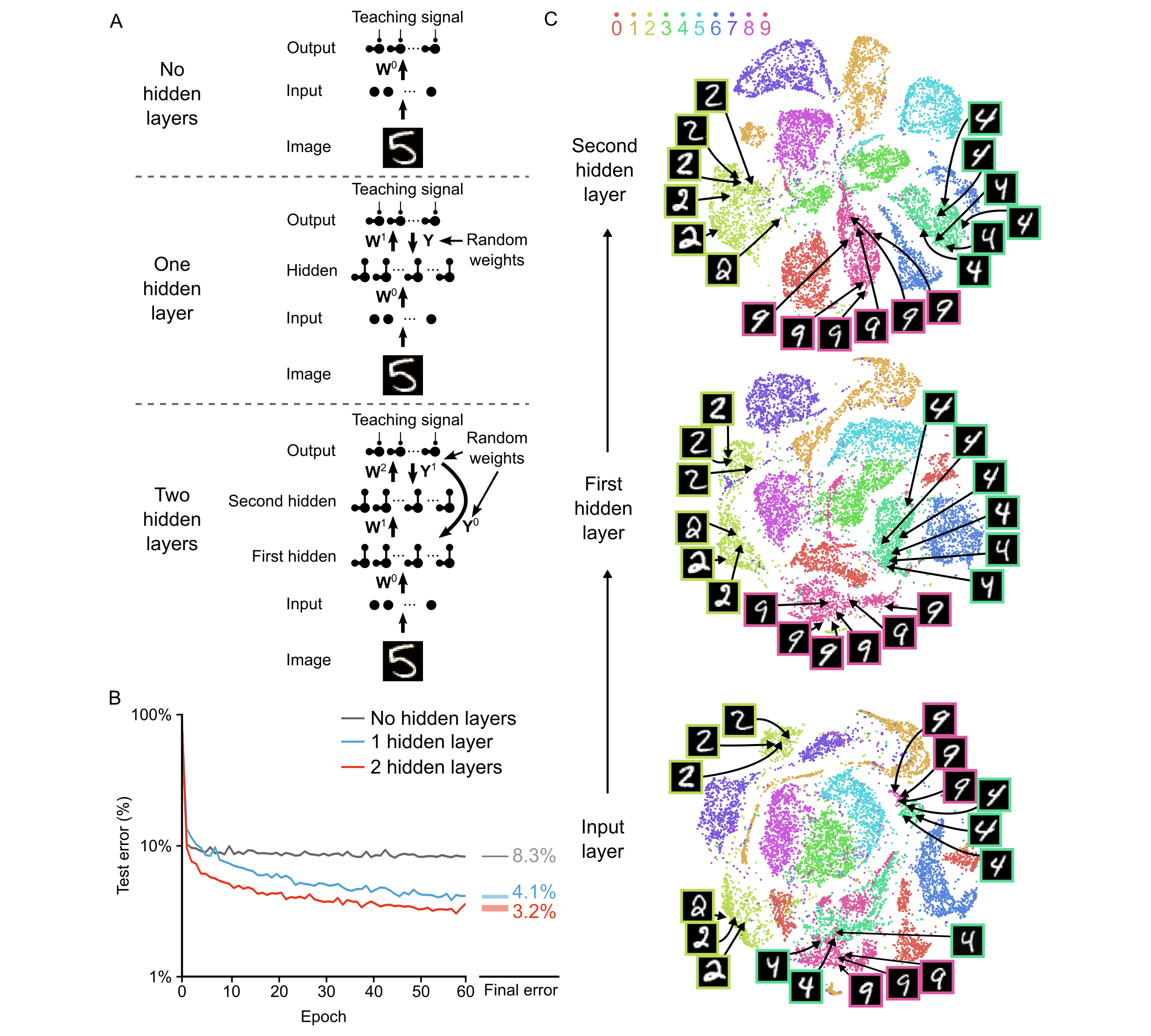}
\caption{\textbf{Improvement of learning with hidden layers.}}
\par}
\medskip
\small
(\textbf{A}) Illustration of the three networks used in the simulations. \textit{Top}: a shallow network with only an input layer and an output layer. \textit{Middle}: a network with one hidden layer. \textit{Bottom}: a network with two hidden layers. Both hidden layers receive feedback from the output layer, but through separate synaptic connections with random weights $\bm{Y}^0$ and $\bm{Y}^1$.
(\textbf{B}) Plot of test error (measured on 10,000 MNIST images not used for training) across 60 epochs of training, for all three networks described in \textbf{A}. The networks with hidden layers exhibit deep learning, because hidden layers decrease the test error. \textit{Right}: Spreads (min -- max) of the results of repeated weight tests ($n=20$) after 60 epochs for each of the networks. Percentages indicate means (two-tailed t-test, 1-layer vs. 2-layer: $t_{38}=197.11$, $P_{38}=2.5\times10^{-58}$; 1-layer vs. 3-layer: $t_{38}=238.26$, $P_{38}=1.9\times10^{-61}$; 2-layer vs. 3-layer: $t_{38}=42.99$, $P_{38}=2.3\times10^{-33}$, Bonferroni correction for multiple comparisons). (\textbf{C}) Results of t-SNE dimensionality reduction applied to the activity patterns of the first three layers of a two hidden layer network (after 60 epochs of training). Each data point corresponds to a test image shown to the network. Points are color-coded according to the digit they represent. Moving up through the network, images from identical categories are clustered closer together and separated from images of different catefories. Thus the hidden layers learn increasingly abstract representations of digit categories.
\label{fig:F5}
\end{figure}
\clearpage

\begin{figure}[!htb]
{\centering
\includegraphics[width=\textwidth]{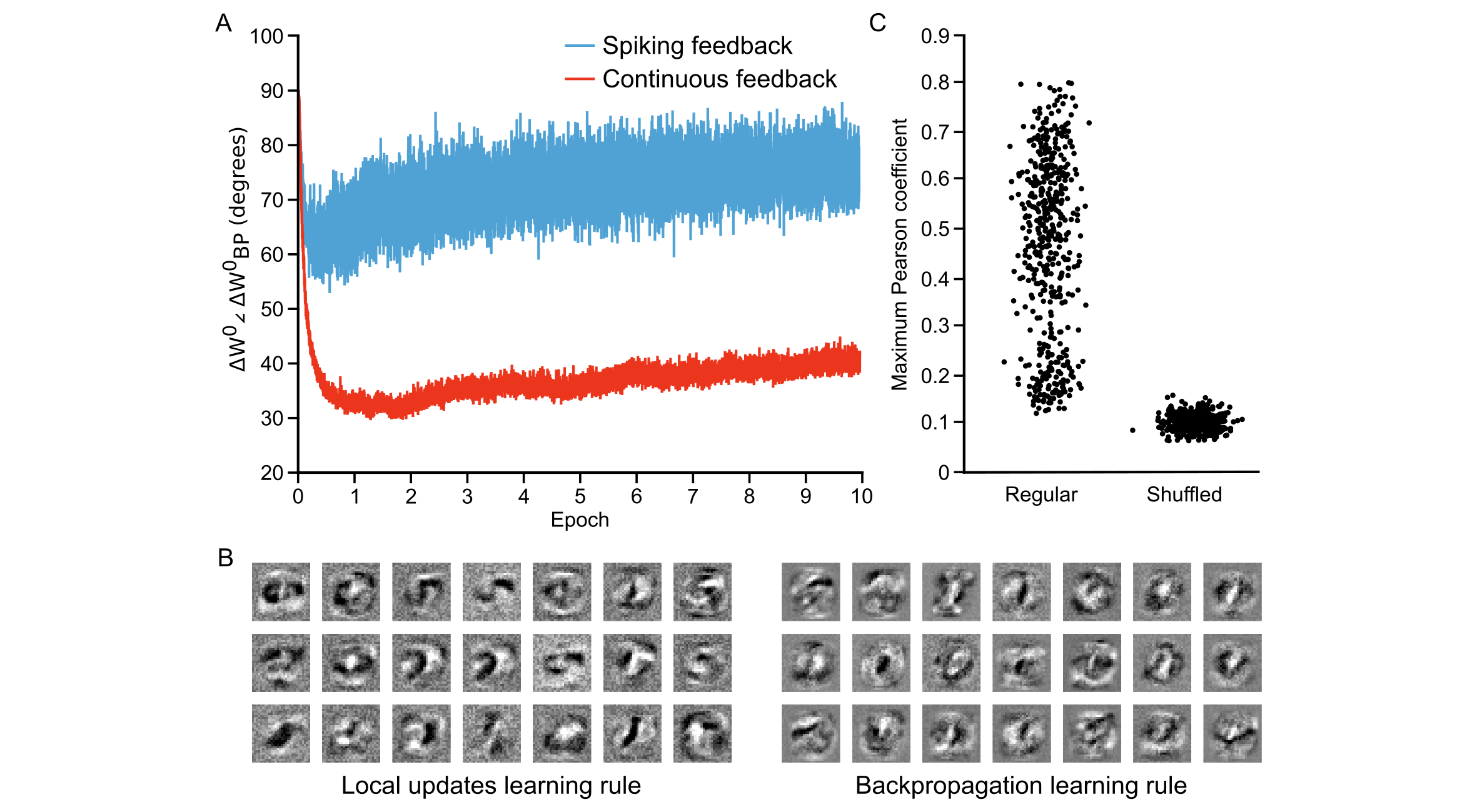}
\caption{\textbf{Approximation of backpropagation with local learning rules.}}
\par}
\medskip
\small
(\textbf{A}) Plot of the angle between weight updates prescribed by our local update learning algorithm compared to those prescribed by backpropagation of error, for a one hidden layer network over 10 epochs of training (each point on the horizontal axis corresponds to one image presentation). Data was time-averaged using a sliding window of 100 image presentations. When training the network using the local update learning algorithm, feedback was sent to the hidden layer either using spiking activity from the output layer units (blue) or by directly sending the spike rates of output units (red). The angle between the local update $\Delta \bm{W}^0$ and backpropagation weight updates $\Delta \bm{W}^0_{BP}$ remains under $90^\circ$ during training, indicating that both algorithms point weight updates in a similar direction.
(\textbf{B}) Examples of hidden layer receptive fields (synaptic weights) obtained by training the network in \textbf{A} using our local update learning rule (left) and backpropagation of error (right) for 60 epochs.
(\textbf{C}) Plot of correlation between local update receptive fields and backpropagation receptive fields. For each of the receptive fields produced by local update, we plot the maximum Pearson correlation coefficient between it and all 500 receptive fields learned using backpropagation (Regular). Overall, the maximum correlation coefficients are greater than those obtained after shuffling all of the values of the local update receptive fields (Shuffled).
\label{fig:F6}
\end{figure}
\clearpage

\begin{figure}[!htb]
{\centering
\includegraphics[width=\textwidth]{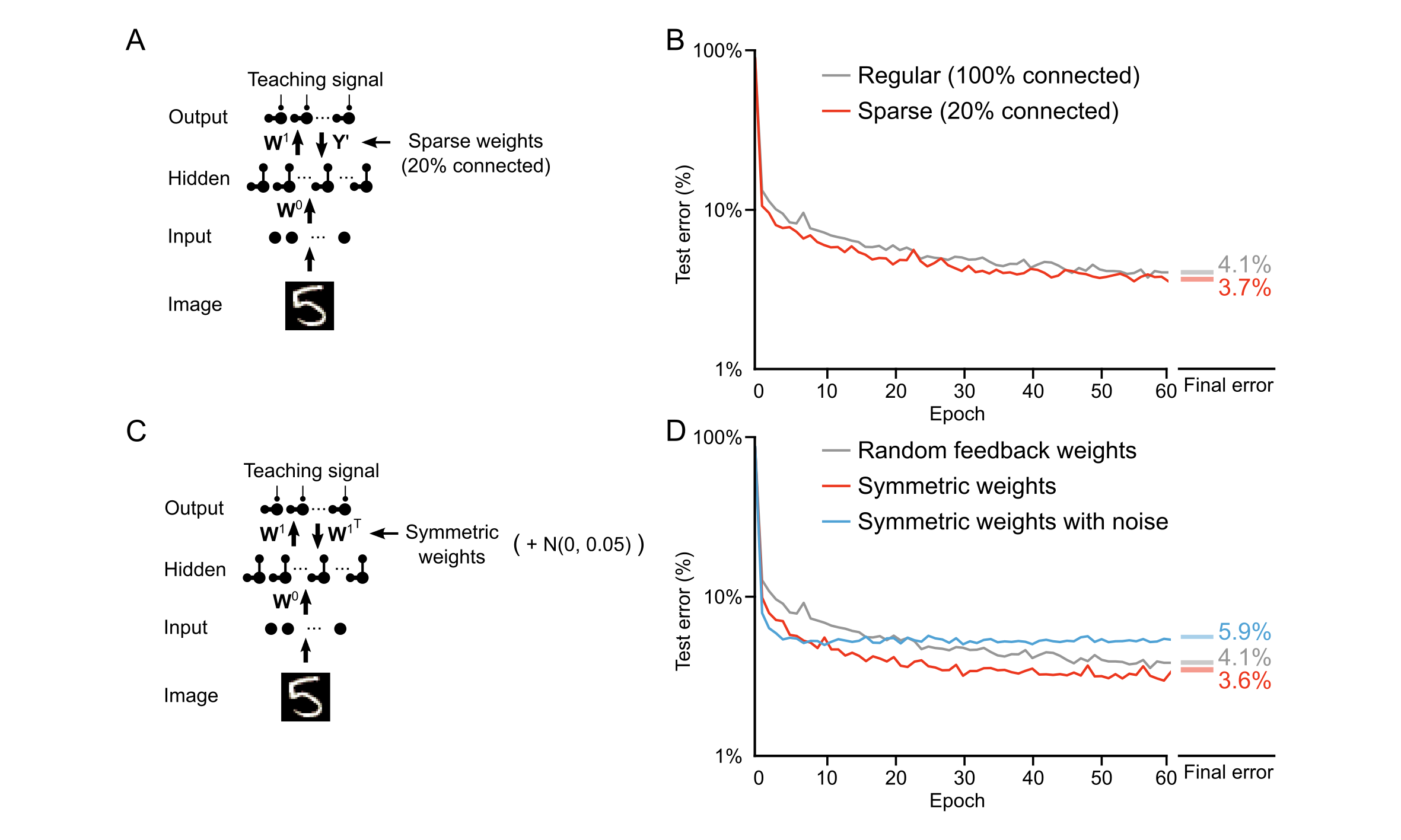}
\caption{\textbf{Conditions on feedback synapses for effective learning.}}
\par}
\medskip
\small
(\textbf{A}) Diagram of a one hidden layer network trained in \textbf{B}, with 80\% of feedback weights set to zero. The remaining feedback weights $\bm{Y}'$ were multiplied by 5 in order to maintain a similar overall magnitude of feedback signals.
(\textbf{B}) Plot of test error across 60 epochs for our standard one hidden layer network (gray) and a network with sparse feedback weights (red). Sparse feedback weights resulted in improved learning performance compared to fully connected feedback weights. \textit{Right}: Spreads (min -- max) of the results of repeated weight tests ($n=20$) after 60 epochs for each of the networks. Percentages indicate mean final test errors for each network (two-tailed t-test, regular vs. sparse: $t_{38}=16.43$, $P_{38}=7.4\times10^{-19}$).
(\textbf{C}) Diagram of a one hidden layer network trained in \textbf{D}, with feedback weights that are symmetric to feedforward weights $\bm{W}^1$, and symmetric but with added noise. Noise added to feedback weights is drawn from a normal distribution with variance $\sigma = 0.05$.
(\textbf{D}) Plot of test error across 60 epochs of our standard one hidden layer network (gray), a network with symmetric weights (red), and a network with symmetric weights with added noise (blue). Symmetric weights result in improved learning performance compared to random feedback weights, but adding noise to symmetric weights results in impaired learning. \textit{Right}: Spreads (min -- max) of the results of repeated weight tests ($n=20$) after 60 epochs for each of the networks. Percentages indicate means (two-tailed t-test, random vs. symmetric: $t_{38}=18.46$, $P_{38}=4.3\times10^{-20}$; random vs. symmetric with noise: $t_{38}=-71.54$, $P_{38}=1.2\times10^{-41}$; symmetric vs. symmetric with noise: $t_{38}=-80.35$, $P_{38}=1.5\times10^{-43}$, Bonferroni correction for multiple comparisons).
\label{fig:F7}
\end{figure}
\clearpage

\begin{figure}[!htb]
{\centering
\includegraphics[width=\textwidth]{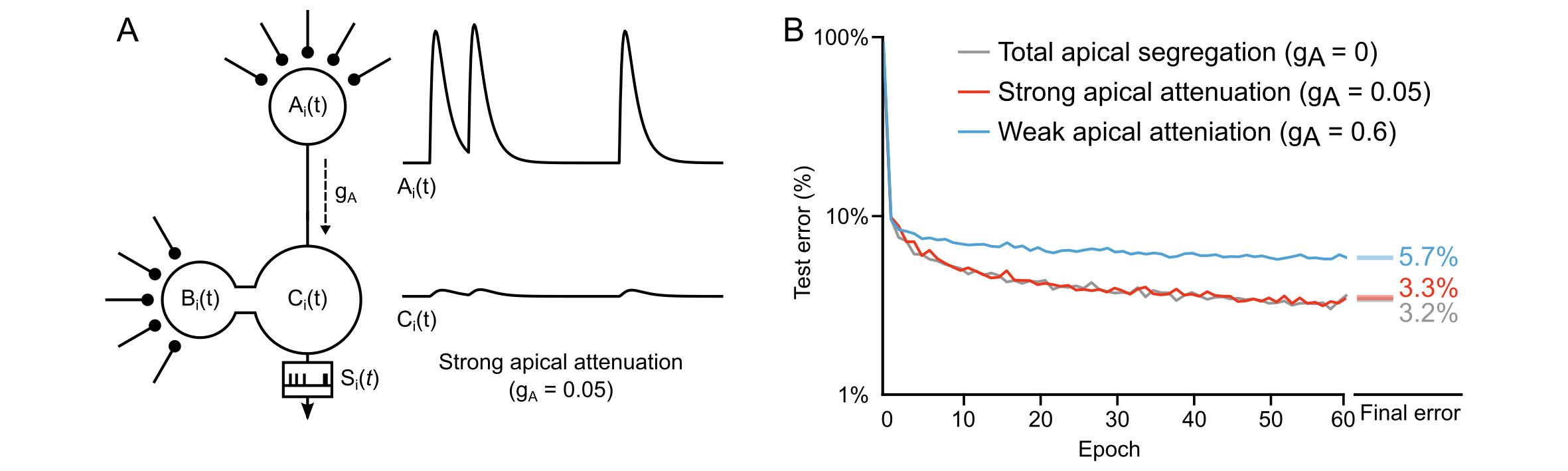}
\caption{\textbf{Importance of dendritic segregation for deep learning.}}
\par}
\medskip
\small
(\textbf{A}) \textit{Left}: Diagram of a hidden layer neuron with an unsegregated apical dendrite that affects the somatic potential $C_i$. $g_A$ represents the strength of the coupling between the apical dendrite and soma. \textit{Right}: Example traces of the apical voltage $A_i$ and the somatic voltage $C_i$ in response to spikes arriving at apical synapses. Here $g_A = 0.05$, so the apical activity is strongly attenuated at the soma.
(\textbf{B}) Plot of test error across 60 epochs of training on MNIST of a two hidden layer network, with total apical segregation (gray), strong apical attenuation (red) and weak apical attenuation (blue). Apical input to the soma did not prevent learning if it was strongly attenuated, but weak apical attenuation impaired deep learning. \textit{Right}: Spreads (min -- max) of the results of repeated weight tests ($n=20$) after 60 epochs for each of the networks. Percentages indicate means (two-tailed t-test, total segregation vs. strong attenuation: $t_{38}=-4.00$, $P_{38}=8.4\times10^{-4}$; total segregation vs. weak attenuation: $t_{38}=-95.24$, $P_{38}=2.4\times10^{-46}$; strong attenuation vs. weak attenuation: $t_{38}=-92.51$, $P_{38}=7.1\times10^{-46}$, Bonferroni correction for multiple comparisons).
\label{fig:F8}
\end{figure}
\clearpage

\section*{Figure Supplements}
\label{sfigures}

\begin{figure}[!htb]
{\centering
\includegraphics[width=0.5\textwidth]{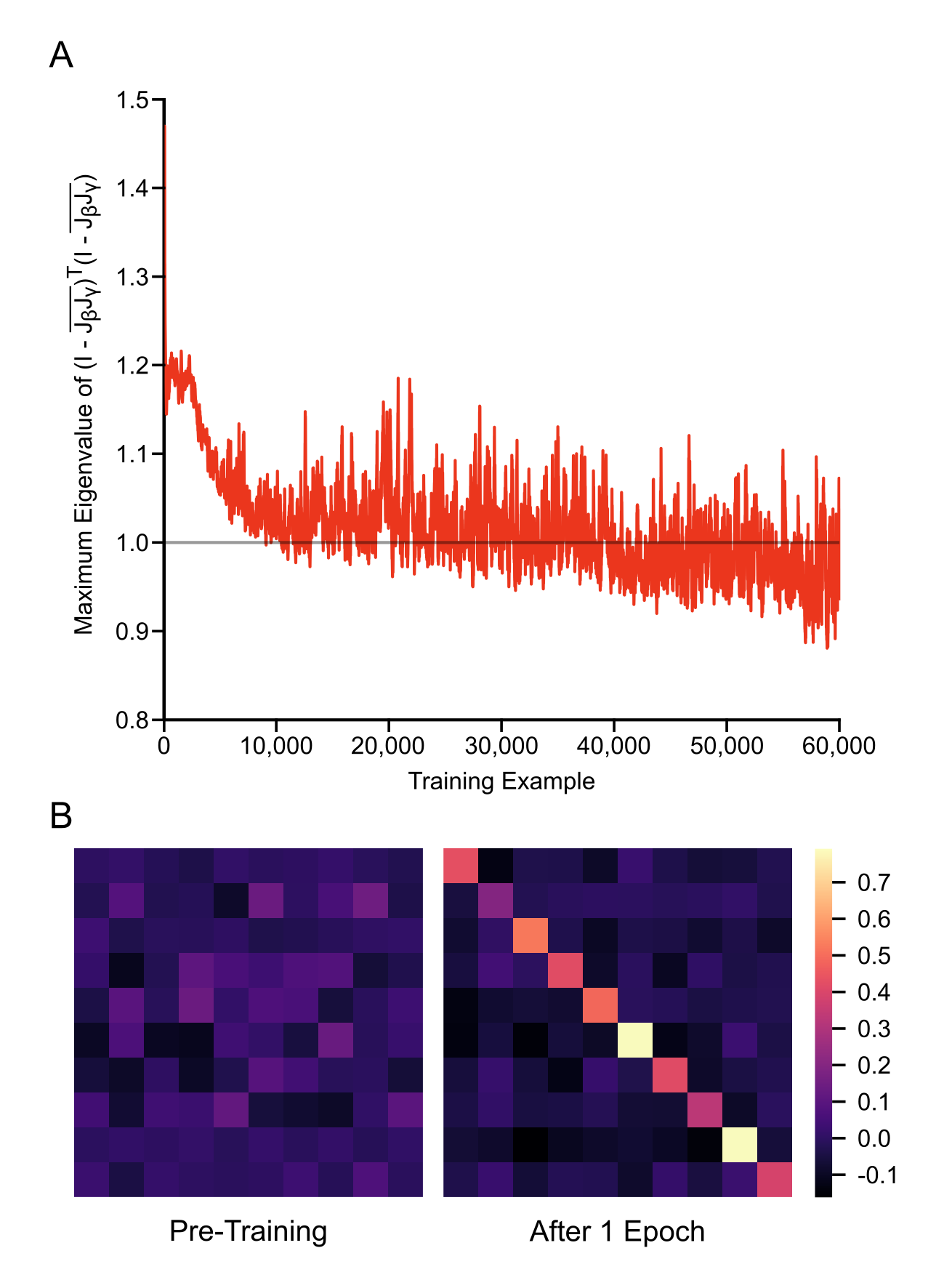}
\captionsetup{labelformat=empty}
\caption{\textbf{Figure 4, Supplement 1. Weight alignment during first epoch of training.}}
\par}
\medskip
\small
(\textbf{A}) Plot of the maximum eigenvalue of $(I - \overline{J_\beta} \overline{J_\gamma})^T (I - \overline{J_\beta} \overline{J_\gamma})$ over 60,000 training examples for a one hidden layer network, where $\overline{J_\beta}$ and $\overline{J_\gamma}$ are the mean feedforward and feedback Jacobian matrices for the last 100 training examples. The maximum eigenvalue of $(I - \overline{J_\beta} \overline{J_\gamma})^T (I - \overline{J_\beta} \overline{J_\gamma})$ drops below 1 as learning progresses, satisfying the main condition for the learning guarantee described in \hyperref[thm1]{Theorem 1} to hold. (\textbf{B}) The product of the mean feedforward and feedback Jacobian matrices, $\overline{J_\beta} \overline{J_\gamma}$, for a one hidden layer network, before training (left) and after 1 epoch of training (right). As training progresses, the network updates its weights in a way that causes this product to approach the identity matrix, meaning that the two matrices are roughly inverses of each other.
\label{fig:F4S1}
\end{figure}
\clearpage

\begin{figure}[!htb]
{\centering
\includegraphics[width=\textwidth]{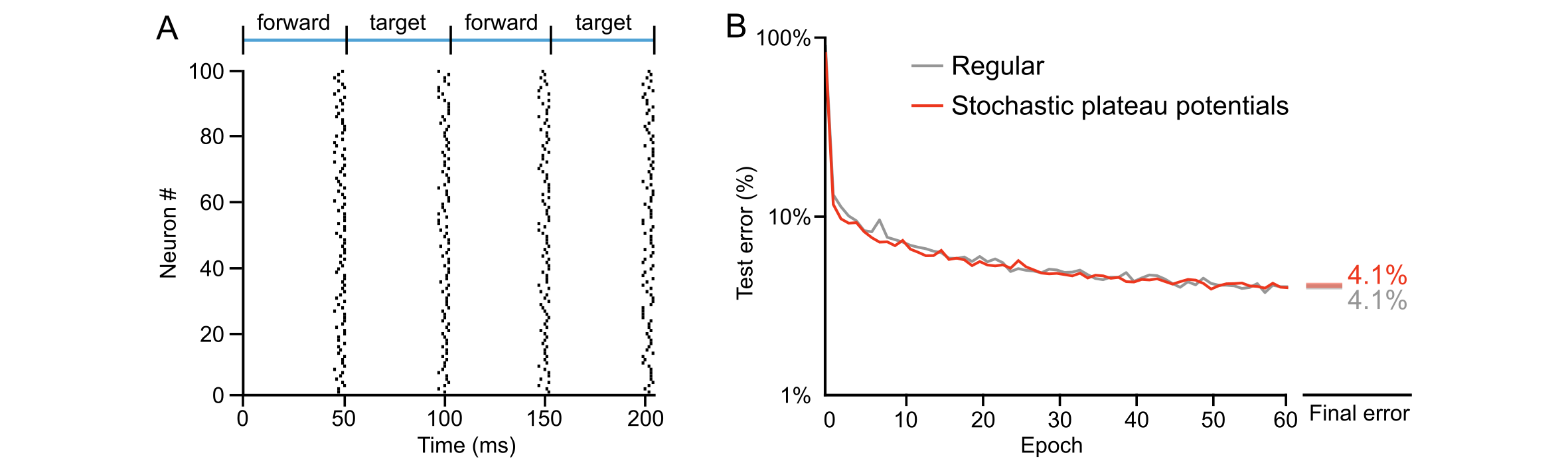}
\captionsetup{labelformat=empty}
\caption{\textbf{Figure 5, Supplement 1. Learning with stochastic plateau times.}}
\par}
\medskip
\small
(\textbf{A}) \textit{Left}: Raster plot showing plateau potential times during presentation of two training examples for 100 neurons in the hidden layer of a network where plateau potential times were randomly sampled for each neuron from a folded normal distribution ($\mu=0, \sigma^2=3$) that was truncated ($\text{max}=5$) such that plateau potentials occurred between 0 ms and 5 ms before the start of the next phase. In this scenario, the apical potential over the last 30 ms was integrated to calculate the plateau potential for each neuron. (\textbf{B}) Plot of test error across 60 epochs of training on MNIST of a one hidden layer network, with synchronized plateau potentials (gray) and with stochastic plateau potentials (red). Allowing neurons to undergo plateau potentials in a stochastic manner did not hinder training performance.
\label{fig:F5S1}
\end{figure}
\clearpage

\begin{figure}[!htb]
{\centering
\includegraphics[width=0.7\textwidth]{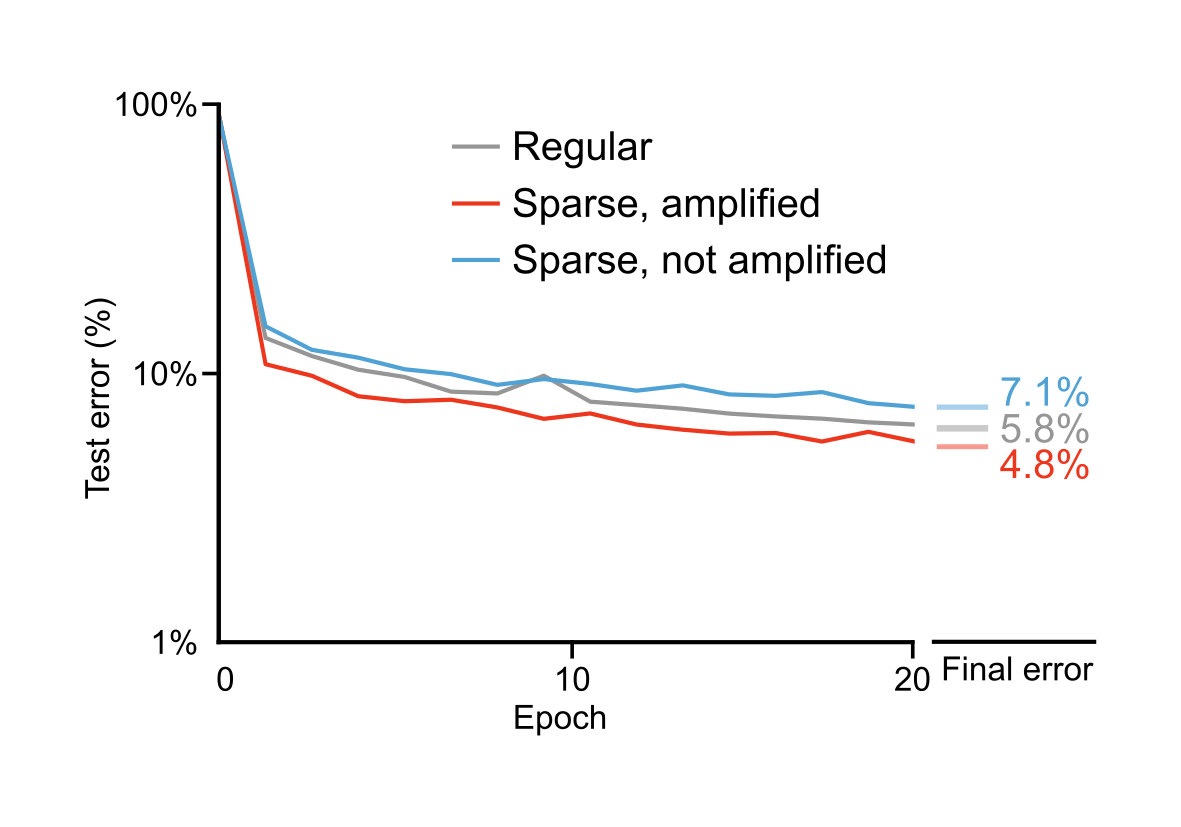}
\captionsetup{labelformat=empty}
\caption{\textbf{Figure 7, Supplement 1. Importance of weight magnitudes for learning with sparse weights.}}
\par}
\medskip
\small
Plot of test error across 20 epochs of training on MNIST of a one hidden layer network, with regular feedback weights (gray), sparse feedback weights that were amplified (red), and sparse feedback weights that were not amplified (blue). The network with amplified sparse feedback weights is the same as in \hyperref[fig:F7]{Figure 7A \& B}, where feedback weights were multiplied by a factor of 5. While sparse feedback weights that were amplified led to improved training performance, sparse weights without amplification impaired the network's learning ability. \textit{Right}: Spreads (min -- max) of the results of repeated weight tests ($n=20$) after 20 epochs for each of the networks. Percentages indicate means (two-tailed t-test, regular vs. sparse, amplified: $t_{38}=44.96$, $P_{38}=4.4\times10^{-34}$; regular vs. sparse, not amplified: $t_{38}=-51.30$, $P_{38}=3.2\times10^{-36}$; sparse, amplified vs. sparse, not amplified: $t_{38}=-100.73$, $P_{38}=2.8\times10^{-47}$, Bonferroni correction for multiple comparisons).
\label{fig:F7S1}
\end{figure}
\clearpage

\end{document}